\documentclass{vldb}
\usepackage{graphicx}
\usepackage{balance}  % for  \balance command ON LAST PAGE  (only there!)
\newcommand{\remove}[1]{}
\newcommand{\pvs}{\vspace{-8pt}}
\usepackage{comment}
\usepackage{multirow}
\usepackage{subcaption}

\usepackage{amsthm}
  \newtheorem{theorem}{Theorem}
  \newtheorem{lemma}{Lemma}

  \usepackage{pifont}
\newcommand{\cmark}{\ding{51}}%
\newcommand{\xmark}{\ding{55}}%
\usepackage{xcolor} %Use color
\usepackage{balance}

\vldbTitle{Causal Consistency and Latency Optimality: Friend or Foe?}
\vldbAuthors{Diego Didona, Rachid Guerraoui, Jingjing Wang, and \\Willy Zwaenepoel}
\vldbDOI{-}

\begin{document}

% ****************** TITLE ****************************************

\title{Causal Consistency and Latency Optimality:\\Friend or Foe?}

% possible, but not really needed or used for PVLDB:
%\subtitle{[Extended Abstract]
%\titlenote{A full version of this paper is available as\textit{Author's Guide to Preparing ACM SIG Proceedings Using \LaTeX$2_\epsilon$\ and BibTeX} at \texttt{www.acm.org/eaddress.htm}}}

% ****************** AUTHORS **************************************


\numberofauthors{1} %  in this sample file, there are a *total*

\author{
% 1st. author
\alignauthor
Diego Didona, Rachid Guerraoui, Jingjing Wang, Willy Zwaenepoel\\
       \affaddr{EPFL}\\
       \email{first.last@epfl.ch}
}
% There's nothing stopping you putting the seventh, eighth, etc.
% author on the opening page (as the 'third row') but we ask,
% for aesthetic reasons that you place these 'additional authors'
% in the \additional authors block, viz.
\date{1 March 2018}
% Just remember to make sure that the TOTAL number of authors
% is the number that will appear on the first page PLUS the
% number that will appear in the \additionalauthors section.

\maketitle

\begin{abstract}
Causal consistency is an attractive consistency model for replicated data stores. 
It is provably the strongest model that tolerates partitions, it avoids the long latencies associated with strong consistency, and, especially when using
read-only transactions, it prevents many of the anomalies of weaker consistency models. 
Recent work has shown that causal consistency allows ``latency-optimal'' read-only transactions, that are nonblocking, single-version and single-round in terms of communication.
On the surface, this latency optimality  is  very appealing, as the vast majority of applications are assumed to have read-dominated workloads.

In this paper, we show that such ``latency-optimal'' read-only transactions induce an extra overhead on writes; the extra overhead is so high that performance is actually jeopardized, even in read-dominated workloads.
We show this result from a practical and a theoretical angle. 

First, we present a protocol that implements  ``almost laten- cy-optimal'' ROTs but does not impose on the writes any of the overhead of latency-optimal protocols. In this protocol, ROTs are nonblocking, one version and can be configured to use either two or one and a half rounds of client-server communication. We experimentally show that this protocol not only provides better throughput, as expected, but also surprisingly better latencies for all but the lowest loads and most read-heavy workloads.

Then, we prove that the extra overhead imposed on writes by latency-optimal read-only transactions is inherent, i.e., it is not an artifact of the design we consider, and cannot be avoided by {\em any} implementation of latency-optimal read-only transactions. We show in particular that this overhead grows linearly with the number of clients.
\end{abstract}

\section{Introduction}
Geo-replication is gaining momentum in industry~\cite{Nishtala:2013,Lu:2015,Noghabi:2016,Corbett:2013,Bacon:2017,Calder:2011,actor:2017,DeCandia:2007,Verbitski:2017} and academia~\cite{Kraska:2013,Crooks:2016b,Sovran:2011,Zhang:2013,Corbett:2013,Moniz:2017,Zhang:2015,Zhang:2016,Nawab:2015} as a design choice for large-scale data platforms to meet the strict latency and availability requirements of on-line applications~\cite{Rahman:2017,Terry:2013,Ardekani:2014}. Geo-replication aims to reduce operation latencies, by storing a copy of the data closer to the clients, and to increase availability, by keeping multiple copies of the data at different data centers (DC).

{\bf Causal consistency.} Causal consistency (CC)~\cite{Ahamad:1995} is an attractive consistency model for building geo-replicated data stores. On the one hand, it has an intuitive semantics and avoids many anomalies that are allowed under weaker consistency properties~\cite{Vogels:2009,DeCandia:2007}. On the other hand, it avoids the long latencies incurred by strong consistency (e.g., linearizability and strict serializability)~\cite{Herlihy:1990,Corbett:2013} and tolerates network partitions~\cite{Lloyd:2011}. In fact, CC is provably the strongest consistency that can be achieved in an always-available system~\cite{Mahajan:2011,Attiya:2015}. 
 CC is the target consistency level of many systems~\cite{Lloyd:2011,Lloyd:2013,Du:2013,Du:2014,Almeida:2013,Bravo:2017,Eunomia:2017}, it is used in platforms that support multiple levels of consistency~\cite{Balegas:2016,Li:2014}, and it is a building block for strong consistency systems~\cite{Balegas:2015} and formal checkers of distributed protocols~\cite{Gotsman:2016}.

{\bf Read-only transactions in CC.} High-level operations such as producing a web page often translate to multiple reads from the underlying data store~\cite{Nishtala:2013}. Ensuring that all these reads are served from the same consistent snapshot avoids undesirable anomalies, 

in particular, the following well-known anomaly: Alice removes Bob from the access list of a photo album and adds a photo to it but Bob reads the original permission and the new version of the album~\cite{Lloyd:2011,Lu:2016}.
Therefore, the vast majority of CC systems provide read-only transactions (ROTs) to read multiple items at once from a causally consistent snapshot~\cite{Lloyd:2011,Lloyd:2013,Du:2014,Almeida:2013,Akkoorath:2016,Zawirski:2015}. Large-scale applications are often read-heavy~\cite{Atikoglu:2012,Nishtala:2013,Noghabi:2016,Lu:2015}, and achieving low-latency ROTs becomes a first-class concern for CC systems.
 
Earlier CC ROT designs were either blocking~\cite{Du:2013,Du:2014,Akkoorath:2016,Almeida:2013} or  required multiple rounds of communications to complete~\cite{Lloyd:2011,Lloyd:2013,Almeida:2013}. 
Recent work on the COPS-SNOW system~\cite{Lu:2016} has, however, demonstrated that it is possible to perform causally consistent ROTs in a single round of communication, sending only one version of the keys involved, and in a nonblocking fashion. Because it exhibits these three properties, the COPS-SNOW ROT protocol was termed {\em latency-optimal (LO)}.
The protocol achieves LO by imposing additional processing costs on writes. One could argue that this is a correct tradeoff for the common case of read-heavy workloads, because the overhead affects the minority of operations 
and is thus
to the advantage of the majority of them. This paper sheds a different light on this tradeoff.   

~\noindent{\bf Contributions.} In this paper we show that the extra cost on writes is so high that so-called latency-optimal ROTs in practice exhibit higher latencies than alternative designs, even in read-heavy workloads. This extra cost not only reduces the available processing power, leading to lower throughput, but it also leads to higher resource contention, which results in higher queueing times, and, ultimately, in higher latencies.  
We demonstrate this counterintuitive result from two angles, a practical and a theoretical one.

1) From a practical standpoint, we show how an existing and widely used design of CC can be improved to achieve almost all the properties of a latency-optimal design, without incurring the overhead on writes that latency optimality implies. We implement this improved design in a system that we call Contrarian. Measurements in a variety of scenarios demonstrate that, for all but the lowest loads and the most read-heavy workloads, Contrarian provides better latencies and throughput than an LO protocol.

2) From a theoretical standpoint, we show that the extra cost imposed on writes to achieve LO ROTs is {\em inherent} to CC, i.e., it cannot be avoided by {\em any} CC system that implements LO ROTs. We also provide a lower bound on this extra cost in terms of communication overhead. Specifically, we show that the amount of extra information exchanged on writes potentially grows linearly with the number of clients.

The relevance of our theoretical results goes beyond the scope of CC. In fact, they apply to any consistency model strictly stronger than causal consistency, e.g., linearizability~\cite{Herlihy:1990}. Moreover, our result is relevant also for systems that implement hybrid consistency models that include CC~\cite{Balegas:2016} or that implement strong consistency on top of CC~\cite{Balegas:2015}.

~\noindent{\bf Roadmap.} The remainder of this paper is organized as follows. Section~\ref{sec:model} provides introductory concepts and definitions. Section~\ref{sec:sys} surveys the complexities involved in the implementation of ROTs. Section~\ref{sec:sys:contrarian} presents our Contrarian protocol. Section~\ref{sec:eval} compares Contrarian and an LO design. 
Section~\ref{sec:theory} presents our theoretical results.
Section~\ref{sec:rw} discusses related work. Section~\ref{sec:conclusion} concludes the paper. 
\section{System Model}
\label{sec:model}

\subsection{API}
\label{sec:model:api}
We consider a multi-version key value data store. We denote keys by lower case letters, e.g., $x$, and the versions of their corresponding values by capital letters, e.g., $X$. The key value store provides the following operations:

\pvs
~\\\noindent{\bf X $\gets$ GET(x): } A GET operation returns the value of the item identified by $x$. 
GET may return $\bot$  to show that there is no item yet identified by $x$. 

\pvs
~\\\noindent{\bf PUT(x, X): } A PUT operation creates a new version $X$ of an item identified by $x$. 
If item $x$ does not yet exist, the system creates a new item $x$ with value $X$. 

\pvs
~\\\noindent{\bf{ (X, Y, ...)  $\gets$  ROT (x, y, ...) :}} A ROT returns a vector ($X$, $Y$, ...) of versions of the requested keys ($x$, $y$, ... ).  
A ROT may return $\bot$ to show that some item does not exist.

~\\
In the remainder of this paper we focus on PUT and ROT operations.

\subsection{Causal Consistency}
\label{sec:model:cc}
The \emph{causality order} is a happens-before relationship between any two operations in a given execution~\cite{Lamport:1978,Ahamad:1995}. For any two operations $\alpha$ and $\beta$, we say that $\beta$ causally depends on $\alpha$, and we write $\alpha\leadsto \beta$, if and only if at least one of the following conditions holds: $i)$ $\alpha$ and $\beta$ are operations in a single thread of execution, and $\alpha$ happens before $\beta$; $ii)$ $\exists x, X$ such that $\alpha$ creates version $X$ of key $x$, and $\beta$ reads $X$; $iii)$ $\exists \gamma$ such that $\alpha\leadsto \gamma$ and $\gamma\leadsto \beta$. If $\alpha$ and $\beta$ are two PUTs with values $X$ and $Y$ respectively, then (with a slight abuse of notation) we also say $Y$ causally depends on $X$, and we write $X\leadsto Y$. 

A \emph{causally consistent} key value store respects the causality order. Intuitively, if a client reads $Y$ and $X \leadsto Y$, then any subsequent read performed by the client on $x$ returns either $X$ or a newer version. I.e., the client cannot read $X' : X' \leadsto X$. A ROT operation returns item versions from a \emph{causally consistent snapshot}~\cite{Mattern89,Lloyd:2011}: if a ROT returns $X$ and $Y$ such that $X\leadsto Y$, then there is no $X'$ such that $X\leadsto X'\leadsto Y$. 

To circumvent trivial implementations of causal consistency, we require that a value, once written, becomes \emph{eventually visible}, meaning that it is available to be read by all clients after some finite time~\cite{Bailis:2013}.

Causal consistency does not establish an order among concurrent (i.e., not causally related) updates on the same key. Hence, different replicas of the same key might diverge and expose different values~\cite{Vogels:2009}. We consider a system that eventually converges: if there are no further updates, eventually all replicas of any key take the same value, for instance using the last-writer-wins rule~\cite{Thomas:1979}.

Hereafter, when we use the term causal consistency, eventual visibility and convergence are implied.

\subsection{Partitioning and Replication}
\label{sec:model:model}

We target key value stores where the data set is split into $N>1$ partitions. Each key is deterministically assigned to a partition by a hash function. A PUT(x, X) is sent to the partition that stores x. For a ROT(x, ...) a read request is sent to all partitions that store keys in the specified key set.

Each partition is replicated at $M \geq 1$ DCs. Our theoretical and  practical results hold for both single and replicated DCs. In the case of replication, we consider a multi-master design, i.e., all replicas of a key accept PUT operations.

\section{Challenges in Implementing\\Read-only Transactions}
\label{sec:sys}

\begin{comment}
\subsection{Latency-optimal ROTs.} We adopt the same definition as in the original formulation of the SNOW paper~\cite{Lu:2016}, which introduced the concept of latency-optimal ROTs. ROTs are {\em one-shot}, i.e., they can be issued in a single parallel batch, with individual reads within a ROT targeting the partition corresponding to the key being read~\cite{Lu:2016}.
An implementation provides latency-optimal ROTs if it satisfies three properties: one-version, one-round and nonblocking. We now informally describe these properties. A more formal definition is deferred to \S~\ref{sec:theory}. 

\pvs
~\\\noindent{\bf Nonblocking} requires that servers perform a ROT without blocking for any external event (e.g., the acquisition of a lock or the receipt of a message). {\bf One-round} requires that a ROT is served in two steps of communications: from client to servers to invoke the ROT; and from servers to client to  send the reply.  {\bf One-version} requires that servers return to clients only one version of each requested key.
\end{comment}

\begin{figure}[t!]
\centering
        \includegraphics[scale=.25]{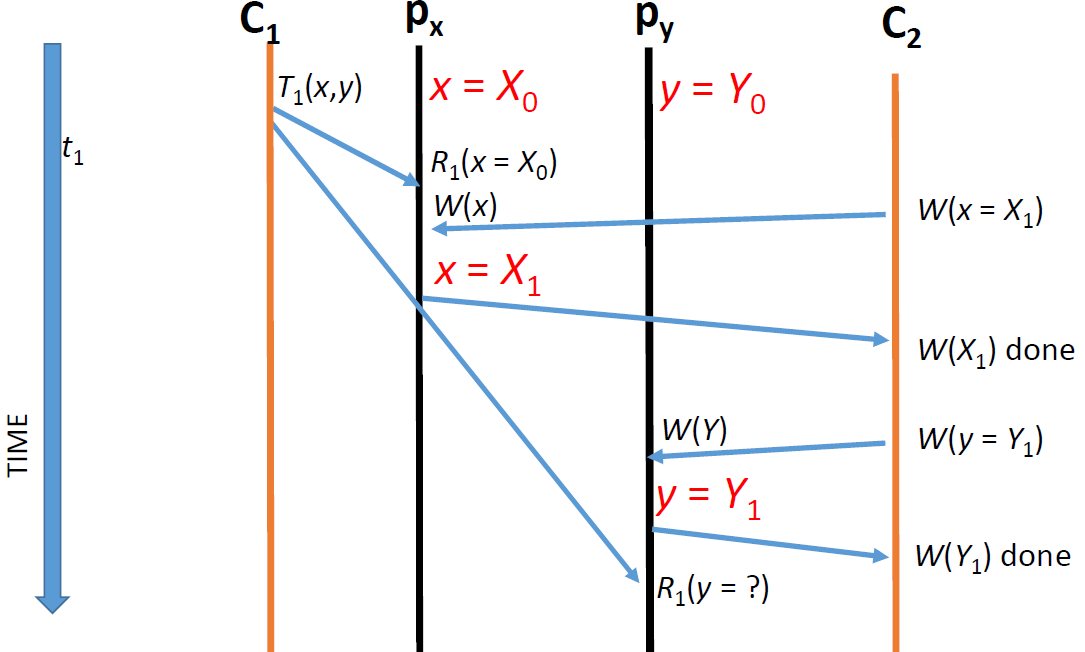}   
        \caption{Challenges in implementing CC ROTs. $C_1$ issues $ROT(x,y)$.  If $T_1$ returns $X_0$ to $C_1$, then $T_1$ cannot return $Y_1$ because there is $X_1$ such that $X_0\leadsto X_1\leadsto Y_1$.}
         \label{fig:sys:cc}
\end{figure}

\noindent{\bf Single DC case.} Even in a single DC, partitions involved in a ROT cannot simply return the most recent version of a requested key if one wants to ensure that a ROT observes a {\color{black}causally consistent snapshot}.  Consider the scenario of Figure~\ref{fig:sys:cc}, with two keys $x$ and $y$, with initial values $X_0$ and $Y_0$, and residing on partitions $p_x$ and $p_y$, respectively. Client $C_1$ performs a ROT on keys $x$ and $y$, and client $C_2$ performs a PUT on $x$ with value $X_1$ and later a PUT on $y$ with value $Y_1$. By asynchrony, the read on x by $C_1$ arrives at $p_x$ before the PUT by $C_2$ on $x$, and the read by $C_1$ on y arrives at $p_y$ after the PUT by $C_2$ on y. Clearly, $p_y$ cannot return $Y_1$ to $C_1$, because a snapshot consisting of $X_0$ and $Y_1$, with $X_0 \leadsto X_1 \leadsto Y_1$ violates the causal consistency property for snapshots (see Section~\ref{sec:model:cc}).

COPS~\cite{Lloyd:2011} presented the first solution to this problem. It encodes causality as direct dependencies of the form ``version $Y$ of $y$ depends on version $X$ of $x$'', stored with $Y$, and ``client $C$ has established a dependency on version $X$ of $x$'', stored with $C$. These dependencies are passed around as necessary to maintain causality. %In our example, 
COPS solves the aforementioned challenge as follows.
when $C_1$ performs its ROT, in the first round of the protocol, $p_x$ and $p_y$ return the most recent version of $x$ and $y$, $X_0$ and $Y_1$. Partition $p_y$ also returns to $C_1$ the dependency ``$Y_1$ depends on $X_1$''. From this piece of information $C_1$ can determine that $X_0$ and $Y_1$ do not form a {\color{black}causally consistent snapshot}. Thus, in the second round of the protocol, $C_1$ requests from $p_x$ a more recent version of $x$ to have a causally consistent snapshot,
%such that there is no violation of the causal snapshot property, 
in this case $X_1$. This protocol is nonblocking, but requires (potentially) two rounds of communication and two versions of key(s) being communicated. Eiger~\cite{Lloyd:2013} improves on this design by using less meta-data, but maintains the potentially two-round, two-version implementation of ROTs.

In later designs for CC systems~\cite{Du:2013,Du:2014,Akkoorath:2016}, direct dependencies were abandoned in favor of timestamps, which provide a more compact and efficient encoding of causality. To maintain causality, a timestamp is associated with every version of every data item. Each client and each partition also maintain the highest timestamp they have observed. When performing a PUT, a client sends along its timestamp. The timestamp of the newly created version is then one plus the maximum between the client's timestamp and the partition's timestamp, thus encoding causality. After completing a PUT, the partition replies to the client with this new version's timestamp. To implement ROTs it then suffices to pick a timestamp for the snapshot, and send it with the ROTs to the partitions. A partition first makes sure that its timestamp has caught up to the snapshot timestamp \cite{Akkoorath:2016}. This ensures that later a version cannot be created with a lower timestamp than the snapshot timestamp. Then, the partition returns the most recent key values with a timestamp smaller than or equal to the snapshot timestamp. 
%NB: returning the version timestamp is not enough in Contrarian

The snapshot timestamp is %commonly   
 picked by a transaction {\em coordinator}~\cite{Du:2014,Akkoorath:2016}. Any server can be the coordinator of a ROT. Thus, the client contacts the coordinator, the coordinator picks the timestamp, and the client or the coordinator then sends this timestamp along with the keys to be read to the partitions. The client provides the coordinator with the last observed timestamp, and the coordinator picks the transaction timestamp as the maximum of the client's timestamp and its own.  Observe that, in general, the client cannot pick the snapshot timestamp itself, because the timestamp may be arbitrarily far behind, compromising eventual visibility. 

Timestamps may be generated by logical or by physical clocks. Returning to our example of Figure 1, assume that the logical clocks at $C_1$ and $C_2$ are initially 0, the logical clocks at $p_x$ and $p_y$ are initially 90. the timestamps of $X_0$ and $Y_0$ are 70, and that a transaction coordinator chooses a snapshot timestamp 100. When receiving the read of $C_1$ with snapshot timestamp 100, $p_x$ advances its logical clock to 100, and returns $X_0$. When $p_x$ receives PUT($x, X_1$), it creates $X_1$ with timestamp 101, and returns that value to $C_2$. $C_2$ then sends the PUT($y, Y_1$) to $p_y$ with timestamp 101, and $Y_1$ is created with timestamp 102. When the read of $C_1$ on $y$ arrives with snapshot timestamp 100, $p_y$ uses the timestamps of $Y_0$ and $Y_1$ to conclude that it needs to return $Y_0$, the most recent version with timestamp smaller than or equal to 100. As with COPS, this protocol is nonblocking; unlike COPS, it requires only a single version of each key, but it always requires two rounds of communication \cite{Akkoorath:2016}.

\begin{figure}[t!]
\centering
        \includegraphics[scale=.23]{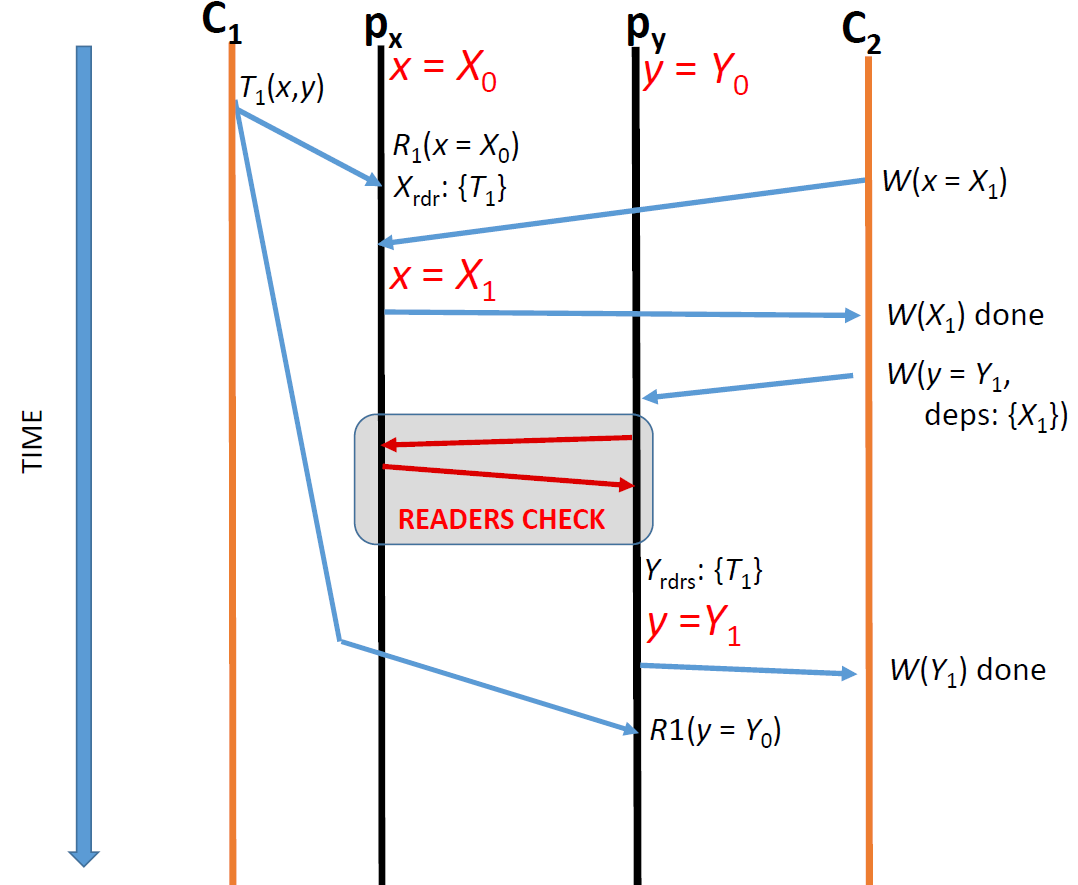}   
        %\caption{How COPS-SNOW solves the challenge in implementing ROTs illustrated in Figure 1}
        \caption{COPS-SNOW design. 
        $C_2$ declares that $Y_1$ depends on $X_0$. Before making $Y_1$ visible, $p_y$  runs a ``readers check'' with $p_x$ and is informed that $T_1$ has observed a snapshot that does not include $Y_1$.}
         \label{fig:sys:cops}
\end{figure}

A further complication arises when (loosely synchronized) physical clocks are used for timestamping~\cite{Du:2014,Akkoorath:2016}, since physical clocks, unlike logical clocks, can only move forward with the passage of time. As a result, in our example, when the read on $p_x$ arrives with snapshot timestamp 100, $p_x$ has to wait until its physical clock advances to 100 before it can return $X_0$. This makes the protocol blocking, in addition to being one-version and two-round.

\begin{figure*}[t!]
\begin{subfigure}[h]{0.5\textwidth}
		\centering
       \includegraphics[scale=0.25]{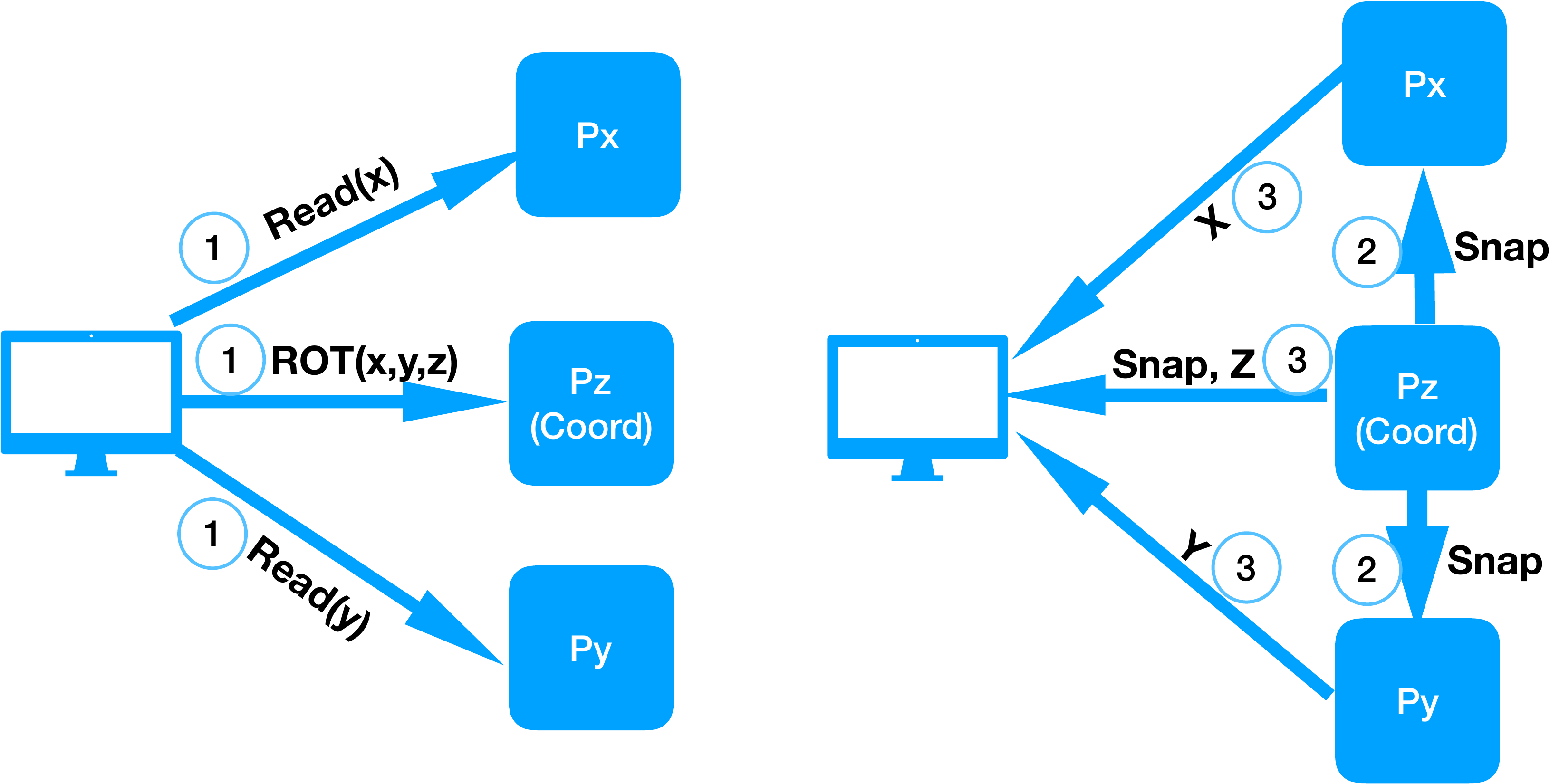}
        \caption{1 1/2 rounds (3 communication steps).}
        \label{fig:contrarian:3h}
    \end{subfigure}
   \hfill 
    \begin{subfigure}[h]{0.5\textwidth}
    \centering
       \includegraphics[scale=0.25]{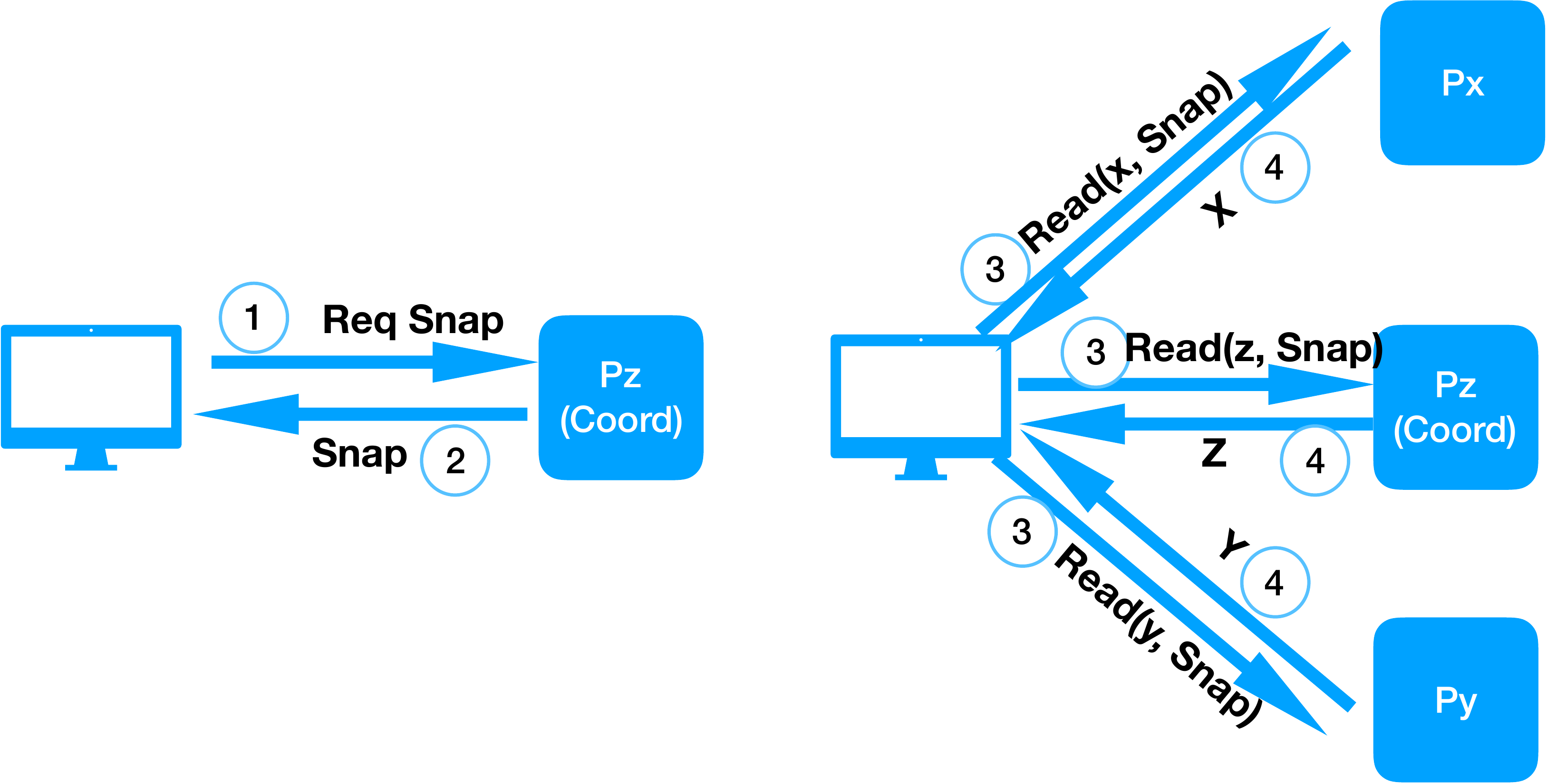}
        \caption{2 rounds (4 communication steps).}
        \label{fig:contrarian:4h}
    \end{subfigure}
\caption{ROT implementation in Contrarian. Numbered circles depict the order of operations. The client always piggybacks on its requests the last snapshot it has seen (not shown), so as to observe monotonically increasing snapshots. Any node involved in a ROT can act as the coordinator of the ROT. Using 1 1/2  rounds reduces the number of communication hops with respect to 2 rounds, at the expenses of more messages exchanged to run a ROT.}\label{fig:contrarian}
\end{figure*}

The question then becomes: does there exist a single-round, single-version, nonblocking protocol for CC ROTs? This question was answered in the affirmative by a follow-up to the COPS and Eiger systems, called COPS-SNOW~\cite{Lu:2016}. Using again the previous example, we depict in Figure~\ref{fig:sys:cops} how the COPS-SNOW protocol works at a high level. Each ROT is given a unique identifier. When a ROT $T_1$ reads $X_0$, $p_x$ records $T_1$ as a reader of $x$. It also records the (logical) time at which the read occurred. On a later PUT on $x$, $T_1$ is added to the ``old readers of $x$'', the set of transactions that have read a version of $x$ that is no longer the most recent version, again together with the logical time at which the read occurred.

When $C_2$ later sends its PUT on $y$ to $p_y$, it includes (as in COPS) that this PUT is dependent on $X_1$. Partition $p_y$ now interrogates $p_x$ as to whether there are old readers of $x$, and, if so, records the old readers of $x$ into the old reader record of $y$, 
together with their logical time. When later the read of $T_1$ on $y$ arrives, $p_y$ finds $T_1$ in the old reader record of $y$. $p_y$ therefore knows that it cannot return $Y_1$. Using the logical time in the old reader record, it returns the most recent version of $y$ before that time, in this case $Y_0$. In the rest of the paper, we refer to this procedure as the {\em readers check}. This protocol is one-round, one-version and nonblocking, and therefore termed {\em latency-optimal}.

This protocol, however, incurs a very high cost on PUTs. We demonstrate this cost by slightly modifying our example. Let us assume that hundreds of ROTs read $X_0$ before the PUT($x, X_1$) (as might well occur with a skewed workload in which x is a hot key). Then all these transactions must be stored as readers and then as old readers of $x$, communicated to $p_y$, and examined by $p_y$ 
on each incoming read from a ROT. Let us further modify the example by assuming that $C_2$ reads other keys from partitions $p_i$ different from $p_x$ and $p_y$ before writing $Y_1$. Because $C_2$ has established a dependency on all the versions it has read, in order to compute the old readers for $y$, $p_y$ needs to interrogate not only $p_x$, but all the other partitions $p_i$.

\pvs
~\\\noindent{\bf Challenges of geo-replication.} Further complications arise in a geo-replicated setting with multiple DCs. We assume that keys are replicated asynchronously, so a new key version may arrive at a DC before its causal dependencies. COPS and COPS-SNOW deal with this situation through a technique called {\em dependency checking}. When a new key version is replicated, its causal dependencies are sent along. Before the new version is installed, the system checks by means of dependency check messages to other partitions that its causal dependencies are present. When its dependencies have been installed in the DC, the new key version can be installed as well. In COPS-SNOW, in addition, the readers check for the new key version proceeds in a remote data center as it does in the data center where the PUT originated. To amortize the overhead, the dependency check and the readers check are performed as a single protocol.

An alternative technique, commonly used with timestamp-based methods, is to use a {\em stabilization protocol}~\cite{Babaoglu:1993,Du:2014,Akkoorath:2016}. Variations exist, but in general each data center establishes a cutoff timestamp below which it has received all remote updates. Updates with a timestamp lower than this cutoff can then be installed. Stabilization protocols are cheaper to implement than dependency checking~\cite{Du:2014}, but they lead to a complication in making ROTs nonblocking, in that one needs to make sure that the snapshot timestamp assigned to a ROT is below the cutoff timestamp, so that there is no blocking upon reading.% We explain how Contrarian addresses this challenge in the next section.

\section{Contrarian: An efficient \\but not latency-optimal design}
\label{sec:sys:contrarian}

We now present Contrarian, a protocol that implements almost all the properties of latency-optimal ROTs, without incurring the overhead that stems from latency-optimal ROTs, thereby providing low latency, resource efficiency and high throughput.

Our goal is not to propose a radically new design of CC. Rather, we aim to show how an existing and widely employed non-latency optimal design can be improved to achieve {\em almost all} the desirable properties of latency optimality without incurring the overhead that inherently results from achieving {\em all} of them (as we demonstrate in Section~\ref{sec:theory}). 

Contrarian builds on the aforementioned coordinator-based design of ROTs and on the stabilization protocol-based approach (to determine visibility of remote items ) in the geo-replica- ted setting.  These characteristics, all or in part, lie at the core of many state-of-the-art systems, like Orbe~\cite{Du:2013}, GentleRain~\cite{Du:2014}, Cure~\cite{Akkoorath:2016} and CausalSpartan~\cite{CausalSpartan}. The improvements we propose in Contrarian, thus, can be employed to improve  the design of these and similar systems.

\pvs
~\\\noindent{\bf Properties of ROTs.} 
Contrarian's ROT protocol runs in 1 1/2 rounds, is one-version, and nonblocking. 
While Contrarian sacrifices a half round in latency compared to the theoretically LO protocol,  it retains the low cost of PUTs as in other non-LO designs.

Contrarian implements ROTs in 1 1/2 rounds of communication, 
by one-round trip between the client and the partitions (one of which is chosen as the coordinator) with an extra hop from the coordinator to the partitions.
As shown in Figure~\ref{fig:contrarian}, this design requires only three communication steps instead of four as the classical coordinator-based approach described in Section 3. 
Contrarian reduces the communication hops to improve latency at the expense of generating more messages to serve a ROT with respect to a 2-round approach. 
As we shall see in Section~\ref{sec:eval:design}, this leads to a slight throughput loss. Contrarian can be configured to run ROTs with 2 rounds (even on a per-ROT basis) to maximize throughput.

Contrarian achieves the one-version property because partitions read the freshest version within the snapshot proposed by the coordinator.

Contrarian implements nonblocking ROTs by using logical clocks.  In the single-DC case, logical clocks allow a partition to move its local clock's value to the snapshot timestamp of an incoming ROT, if needed. Hence, ROTs can be served without blocking (as described in Section 3).

We now describe how Contrarian implements geo-replica- tion and retains the nonblocking property in that setting.

\pvs
~\\\noindent{\bf Geo-replication.} Similarly to Cure~\cite{Akkoorath:2016},  Contrarian uses dependency vectors to track causality, and employs a stabilization protocol to determine a cutoff vector $CV$ in a DC (rather than a cutoff timestamp as discussed earlier). Every partition maintains a version vector $VV$ with one entry per DC. $VV[m]$ is the timestamp of the latest version created by the partition, where $m$ is the index of the local DC. $VV[i], i\neq m$, is the timestamp of the latest update received from the replica in the $i-$th DC. A partition sends a heartbeat message with its current clock value to its replicas if it does not process a PUT for a given amount of time.

Periodically, the partitions within $DC_m$ exchange their $VV$s and compute the aggregate minimum vector, called Global Stable Snapshot ($GSS$). The GSS represents a lower bound on the snapshot of remote items that have been installed by {\em every} node in $DC_m$. The GSS is exchanged between clients and partitions upon each operation to update their views of the snapshot installed in $DC_m$.
 
Items track causal dependencies by means of dependency vectors $DV$, with one entry per DC. If $X.DV[i] = t$, then $X$ (potentially) causally depends on all the items originally written in $DC_i$ with a timestamp up to $t$. $DV[s]$, where $s$ is the source replica, is the timestamp of $X$ and it is enforced to be higher than any other entry in $DV$ upon creation of $X$, to reflect causality. The remote entries of the GSS are used to build the remote entries of $DV$ of newly created items. $X$ can be made visible to clients in a remote $DC_r$ if $X.DV$ is entry-wise smaller than or equal to the GSS on the server that handles $x$ in $DC_r$. This condition implies that all $X'$s dependencies have already been received in $DC_r$. %\leq GSS[i], i \neq r$, meaning that all the remote dependencies of $X$ have been received in $DC_m$. The dependencies on versions created in $DC_r$ are trivially already satisfied in $DC_r$. %$X.DV[m]$ is, by construction, smaller than the current clock of any partition in $DC_m$. 

The ROT protocol uses a vector $SV$ to encode a snapshot. The local entry of $SV$ is the maximum between the clock at the coordinator and the highest local timestamp seen by the client. The remote entries of $SV$ are given by the maximum between the $GSS$ at the coordinator and the highest $GSS$ seen by the client. An item $Y$ belongs to the snapshot encoded by $SV$ if $Y.DV\leq SV$. This protocol is nonblocking because $i)$ partitions can move the value of their local clock forward to match the local entry of $SV$ and $ii)$ the remote entries of $SV$ correspond to a causally consistent snapshot of remote items that have already been received in the DC.% are already available on all the partitions within the DC. 
\begin{comment}
This implies that, in general, the client cannot pick the snapshot vector for a transaction, because the corresponding snapshot could include remote items that have not been received yet in the DC.
\end{comment}

\pvs
~\\\noindent{\bf Freshness of the snapshots.} The $GSS$ is computed by means of the minimum operator. Because logical clocks on different nodes may advance at different paces, a single laggard node in one DC can keep entries in the $GSS$ from progressing, thus increasing the staleness of the snapshot.  A solution to this problem is to use loosely synchronized physical clocks~\cite{Du:2013,Du:2014,Akkoorath:2016}. However, physical clocks cannot be moved forward to match the timestamp of an incoming ROT, which can jeopardize the nonblocking property~\cite{Akkoorath:2016}. 

To achieve fresh snapshots and preserve nonblocking ROTs, Contrarian uses Hybrid Logical Physical Clocks (HLC)~\cite{Kulkarni:2014}. In brief, an HLC is a logical clock that generates timestamps by taking the maximum between the local physical clock and the highest timestamp seen by the node plus one. On the one hand, HLCs behave like logical clocks, so a server can move its clock forward to match the timestamp of an incoming ROT request, thereby preserving the nonblocking behavior of ROTs. On the other hand, HLCs behave like physical clocks, because they advance even in absence of events and inherit the (loose) synchronicity of the underlying physical clock. Hence, the stabilization protocol identifies fresh snapshots. Importantly, the correctness of Contrarian does not depend on the synchronization of the clocks, and Contrarian preserves its properties even if using plain logical clocks.

Contrarian is not the first CC system that proposes the use of HLCs to generate event timestamps. However, existing systems use HLCs either to avoid blocking PUT operations~\cite{CausalSpartan}, or reduce replication delays~\cite{Eunomia:2017}, or improve the clock synchronization among servers~\cite{Mehdi:2017}. Here, we show how HLCs can be used to implement nonblocking ROTs.

\section{Experimental Study}
\label{sec:eval}

\begin{table*}[t!]
\centering
\scriptsize
\begin{tabular}{|c|c|c|c|}
\hline
{\bf Parameter }                          & {\bf Definition}                                                         & {\bf Value}     & {\bf Motivation}                                                                                      \\ \hline
\multirow{3}{*}{Write/read ratio ({\bf w})  } & \multirow{3}{*}{\#PUTS/(\#PUTs+\#individual reads)} & 0.01               & Extremely read-heavy workload                                                       \\ \cline{3-4} 
                                    &                                                                      & {\bf 0.05}             & Default read-heavy parameter in YCSB~\cite{Cooper:2010}                                                             \\ \cline{3-4} \cline{3-4} 
                                    &                                                                      & 0.1           & Default parameter in COPS-SNOW~\cite{Lu:2016}                                                                         
                                     \\ \hline\hline
Size of a ROT ({\bf p})                   & \# Partitions involved in a ROT     & {\bf 4},8,24     & Application operations span multiple partitions~\cite{Nishtala:2013} \\ \hline\hline
\multirow{3}{*}{Size of values ({\bf b})} & \multirow{3}{*}{Value size (in bytes). Keys take 8 bytes.} & {\bf 8}               & Representative of many production workloads~\cite{Atikoglu:2012,Nishtala:2013,Reda:2017}                                                           \\ \cline{3-4} 
                                    &                                                                      & 128             & Default parameter in COPS-SNOW~\cite{Lu:2016}                                                             \\ \cline{3-4} \cline{3-4} 
                                    &                                                                      & 2048            & Representative of workloads with large items                                                                          
                                     \\ \hline\hline
                     \multirow{3}{*}{Skew in key popularity ({\bf z})} & \multirow{3}{*}{Parameter of the zipfian distribution.} & {\bf 0.99}              & Strong skew typical of many production workloads~\cite{Atikoglu:2012,Balmau:2017}                                                           \\ \cline{3-4} 
                                    &                                                                      & 0.8             & Moderate skew and default in COPS-SNOW~\cite{Lu:2016}                                                             \\ \cline{3-4} 
                                    &                                                                      & 0            & No skew (uniform distribution)~\cite{Balmau:2017}
                                     \\ \hline
\end{tabular}
\caption{Workload parameters considered in the evaluation. The default values are given in bold.}
\label{tab:wkld}
\end{table*}

\subsection{Summary of the results}
\noindent{\bf Main findings.} We show that the resource demands to perform PUT operations in the latency-optimal design are in practice so high that they not only affect the performance of PUTs, but also the performance of ROTs, even with read-heavy workloads. In particular, with the exception of scenarios corresponding to extremely read-heavy workloads and modest loads, where the two designs are comparable, Contrarian achieves ROT latencies that are lower than a latency-optimal design. In addition, Contrarian achieves higher throughput for almost all workloads.

\pvs
~\\\noindent{\bf Lessons learnt.} In light of our experimental findings, we draw three main conclusions.

\pvs
~\\\noindent{$i)$} Overall system efficiency is key to {\em both} low latency {\em and} high throughput. It is fundamental to understand the cost of optimizing an operation on the system even though the optimized operation dominates the workload. 

\pvs
~\\\noindent{$ii)$} The high-level theoretical model of a design may not capture the resource utilization dynamics incurred by the design. While a theoretical model represents a powerful lens to compare and qualitatively analyze designs, the choice of a target design for a system should rely also on a more quantitative analysis, e.g., by means of analytical modeling~\cite{Tay:2010}.

\pvs
~\\\noindent{$iii)$} Ultimately, the optimality of a design is closely related to the target workload as well as target architecture and available computational resources.

\subsection{Experimental environment.} 
\noindent{\bf Implementation and optimizations.} We implement Contrarian, Cure~\footnote{Cure supports an API that is different from Contrarian's~\cite{Akkoorath:2016}. We modify Cure to comply with the model described in Section~\ref{sec:model}.} and the COPS-SNOW design in the same C++ code-base. Clients and servers use Google Protocol Buffer~\cite{protobuf} for communication. We call CC-LO the system that implements the design of COPS-SNOW. We improve its performance over the original design by more aggressive eviction of transactions from the old reader record. Specifically, we garbage-collect a ROT id after 500 msec from its insertion in the readers list of a key (vs the 5 seconds of the original implementation) and we enforce that each readers-check message response contains at most one ROT id per client, i.e., the one corresponding to the most recent ROT of that client. These two optimizations reduce by one order of magnitude the amount of ROT ids exchanged, leading it to approach the lower bound we describe in Section 6.
We use NTP~\cite{ntp} to synchronize clocks in Contrarian and Cure, and the stabilization protocol is run every 5 msec.

\pvs 
~\\\noindent{\bf Platform.} We use 64 machines equipped with 2x4 AMD Opteron 6212 (16 hardware threads) and 130 GB of RAM and running Ubuntu 16.04 with a 4.4.0-89-generic kernel. We consider a data set sharded across 32 partitions. Each partition is assigned to a different server. We consider a single DC scenario and a replicated scenario with two replicas. Machines communicate over a 10Gbps network. 
 
Using only two replicas is a favorable condition for CC-LO, 
Since the readers check has to be performed also for replicated updates in the remote DCs, the corresponding overhead grows linearly with the number of DCs. 
We also note that the overheads of the designs we consider are largely unaffected by the communication latency between replicas, because update replication is asynchronous and happens in the background. Thus, emulating a multi-DC scenario over a local area network suffices to capture the most relevant performance dynamics that depend on (geo-)replication~\cite{Lu:2016}.

\pvs
~\\\noindent{\bf Methodology.} Experiments run for 90 seconds, and clients issue operations in closed loop. We generate different loads for the system by spawning different numbers of client threads (starting from one thread per client machine).  We have run each experiment up to 5 times, with minimal variations between runs, and report the median result. 

\pvs
~\\\noindent{\bf Workloads.} Table~\ref{tab:wkld} summarizes the workload parameters we consider. We use read-heavy workloads, in which clients issue ROTs and PUTs according to a given w/r ratio (w), defined as \#PUT/(\#PUT + \#READ). A ROT reading $k$ keys counts as $k$ READs. ROTs  span a target number of partitions (p), chosen uniformly at random, and read one key per partition. Keys in a partition are chosen according to a zipfian distribution with a given parameter (z). Every partition stores 1M keys, and items have a constant size (b). 

The default workload we consider uses w = 0.05, i.e., the default value for the read-heavy workload in YCSB~\cite{Cooper:2010}; z = 0.99, which is representative of skewed workloads~\cite{Atikoglu:2012}; p = 4, which corresponds to small ROTs (which exacerbate the extra communication cost  in Contrarian); and b = 8, as many production workloads are dominated by tiny items~\cite{Atikoglu:2012}.
  
\pvs
~\\\noindent{\bf Performance metrics.} We focus our study on the latencies of ROTs because, by design, CC-LO favors ROT latencies over PUTs. 
As an aside, in our experiments CC-LO incurs up to one order of magnitude higher PUT latencies than Contrarian. 
For space constraints, we focus on average latencies. We report the 99-th percentile of latencies for a subset of the experiments. We measure the throughput of the systems as the number of PUTs and ROTs per second.

\begin{figure}[b!]
\centering
\includegraphics[scale=0.55]{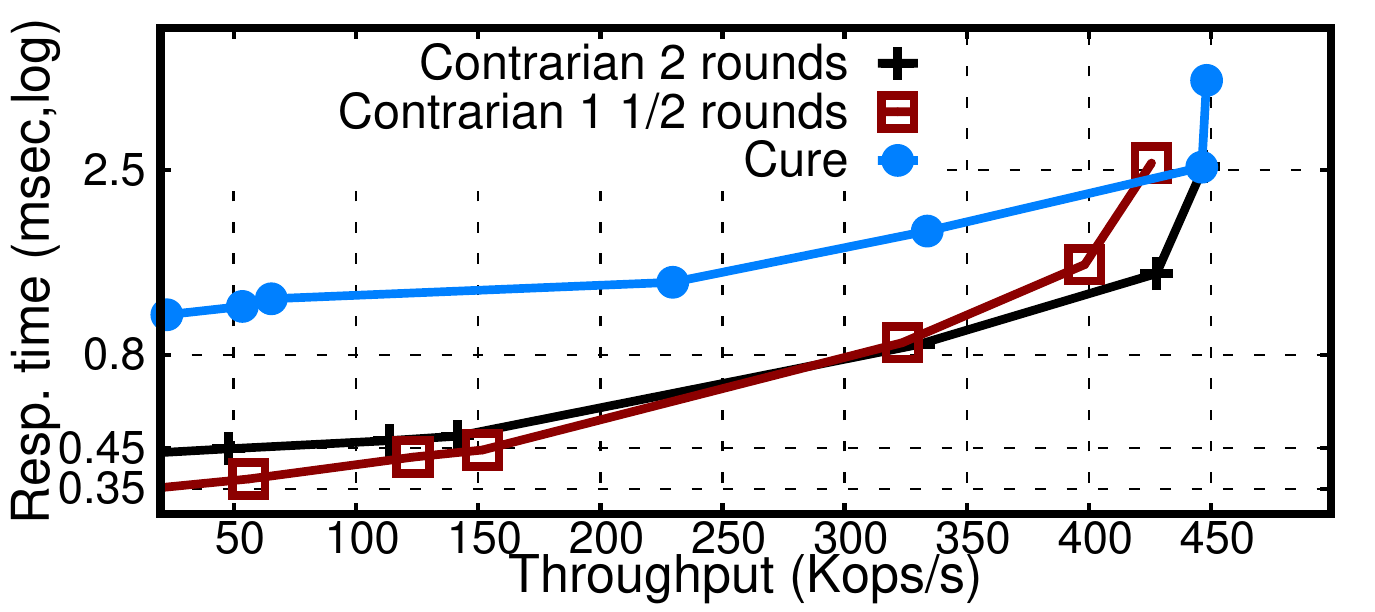}
        \caption{Evaluation of Contrarian's design (2-DC, default workload). Throughput vs average ROT latency (y axis in log scale). Contrarian achieves lower latencies than Cure by means of nonblocking ROTs. Using 1 1/2 rounds of communication reduces latency at low load, but it leads to exchange more messages than using 2 rounds, and hence to a lower maximum throughput (Section~\ref{sec:sys:contrarian}).}\label{fig:eval:contrarian}
\end{figure}

\begin{figure*}[t!]
\begin{subfigure}[h]{0.5\textwidth}
		\centering
       \includegraphics[scale=0.55]{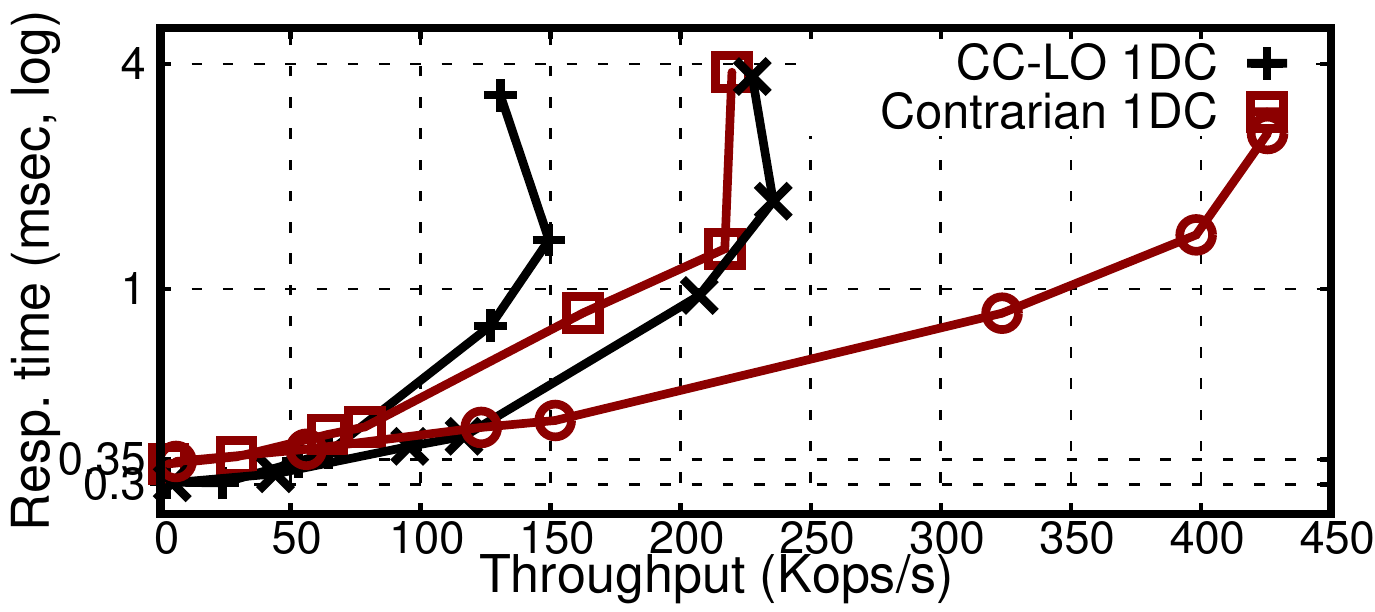}
        \caption{Throughput vs Avg. ROT latency.}
        \label{fig:default:avg}
    \end{subfigure}
    %\hspace{-1cm}
    \begin{subfigure}[h]{0.5\textwidth}
    \centering
       \includegraphics[scale=0.55]{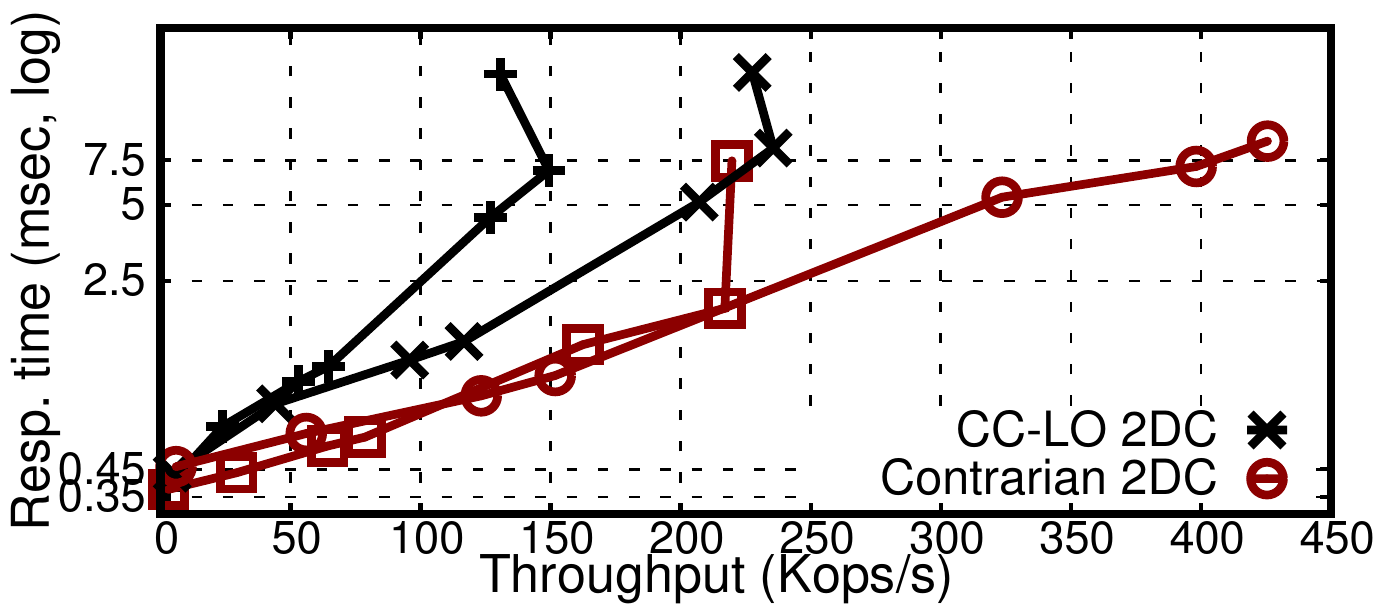}
        \caption{Throughput vs 99-th percentile of ROT latencies.}
        \label{fig:default:99th}
    \end{subfigure}
\caption{ROT latencies (average and 99-th percentile) in Contrarian and CC-LO as a function of the throughput (default workload). The resource contention induced by the extra overhead posed by PUTs in CC-LO affects especially tail latency.}\label{fig:default}
\end{figure*}

\begin{figure}[b!]
\centering
\includegraphics[scale=0.55]{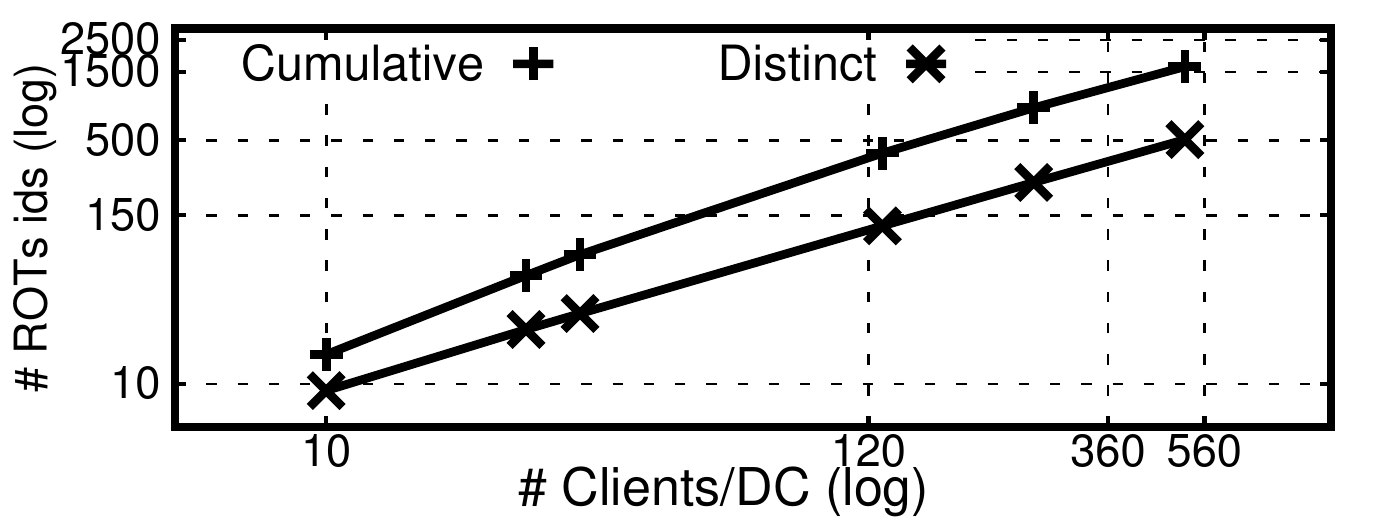}
    \caption{ROT ids collected on average during a readers check in CC-LO (1-DC, default workload). The amount of information exchanged grows linearly with the number of clients, matching the bound stated in Section~\ref{sec:theory}. The average number of servers contacted during a readers check is 12.}
    \label{fig:eval:rdrs}
\end{figure}

\begin{figure*}[t!]
\begin{subfigure}[h]{0.5\textwidth}
		\centering
       \includegraphics[scale=0.55]{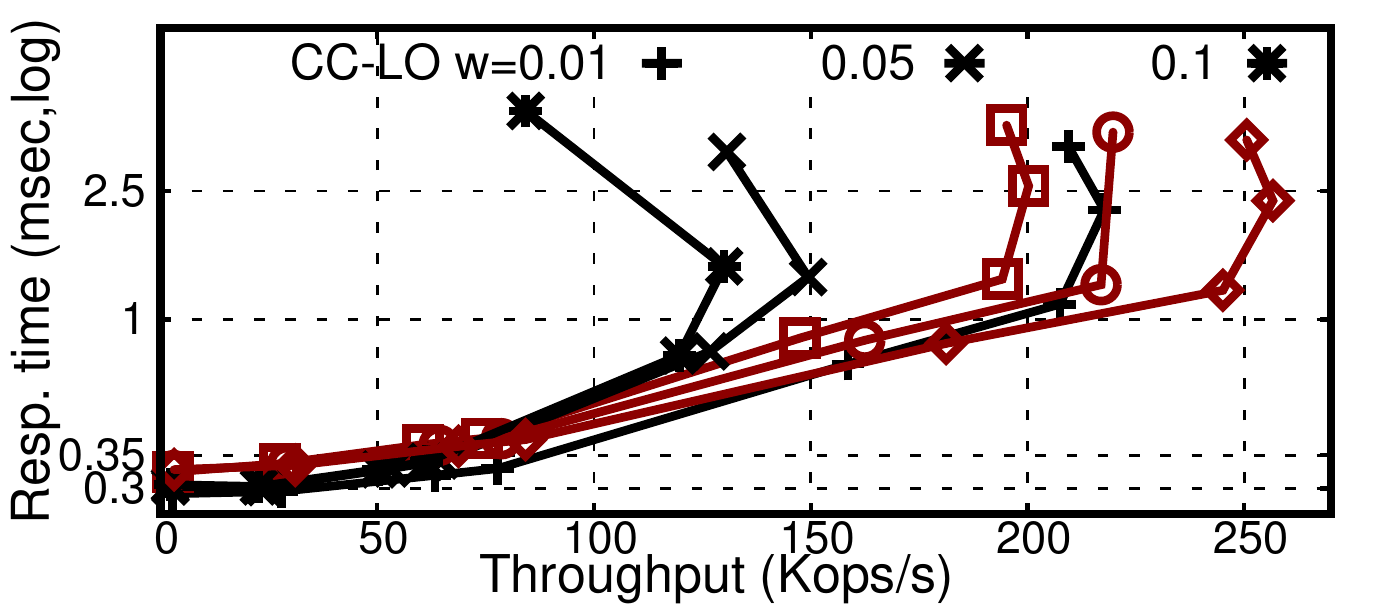}
        \caption{Throughput vs Avg. ROT latency (1 DC).}
        \label{fig:eval:write:1DC}
    \end{subfigure}
\begin{subfigure}[h]{0.5\textwidth}
		\centering
       \includegraphics[scale=0.55]{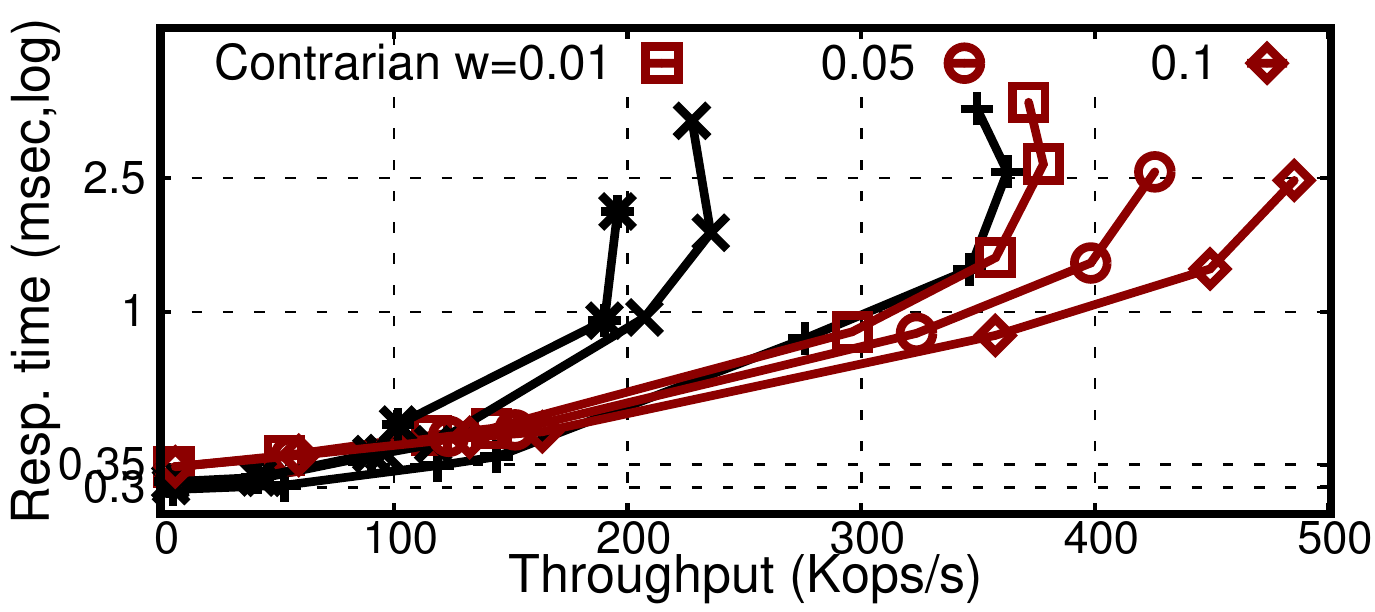}
        \caption{Throughput vs Avg. ROT latency (2 DCs).}
        \label{fig:eval:write:2DC}
    \end{subfigure}    
\caption{Performance with different w/r ratios. Contrarian achieves lower ROT latencies than CC-LO, except at very moderate load and for the most read-heavy workload. Contrarian also consistently achieves higher throughput. Higher write intensities hinder the performance of CC-LO because the readers check is triggered more frequently.}\label{fig:write}
\end{figure*}

\subsection{Contrarian design}
\label{sec:eval:design}
We first evaluate the design of Contrarian, by assessing its improvement over Cure, and by analyzing the behavior of the system when implementing ROTs in 1 1/2 or 2 rounds of communication. 
Figure~\ref{fig:eval:contrarian} compares the three designs given the default workload in 2 DCs.

Contrarian achieves lower latencies than Cure, up to a factor of $\approx$ 3x (0.35 vs 1.0 msec), by implementing nonblocking ROTs. In Cure, the latency of ROTs is affected by clock skew. At low load, the 1 1/2-round version of Contrarian completes ROTs in 0.35 msec vs the 0.45 msec of the 2-round version. The two variants achieve comparable latencies at medium/high load (from 150 to 350 Kops/s). The 2-round version  achieves a higher throughput than the 1 1/2-round version (by 8\% in this case) because it is more resource efficient by requiring fewer messages to run ROTs.

Because we focus on latency more than throughput, hereafter we report results corresponding to the 1 1/2-round version of Contrarian. 

\subsection{Default workload.} 
Figure~\ref{fig:default} reports the performance of Contrarian and CC-LO with the default workload, in the 1-DC and 2-DC scenarios. Figure~\ref{fig:default}(a) reports average latencies, and Figure~\ref{fig:default}(b) reports 99-th percentile. Figure~\ref{fig:eval:rdrs} reports information on the readers check overhead in CC-LO in the single-DC case. 

\pvs
~\\\noindent{\bf Latency.} Figure~\ref{fig:default} (a) shows that Contrarian achieves higher latencies than CC-LO only at very moderate load. Under trivial load conditions ROTs in CC-LO take 0.3 msec on average vs the 0.35 of Contrarian. For the throughput higher than 60 Kops/s in the 1-DC case and than 120 Kops/s in the 2-DC case Contrarian achieves lower latencies than CC-LO. 
These load conditions correspond to roughly 25\%  of Contrarian's peak throughput. That is, CC-LO achieves slightly better latencies than Contrarian only for load conditions 
that correspond to the case where the available resources are severely under-utilized. 

CC-LO achieves worse latencies than Contrarian for nontrivial load conditions because of the overhead caused by the readers check, needed to achieve latency optimality. This overhead induces higher resource utilization, and hence higher contention on physical resources. Ultimately, this leads to higher latencies, even for ROTs.

\pvs
~\\\noindent{\bf Tail latency.} The effect of contention on physical resources is especially visible at the tail of the ROT latencies distribution, as shown in Figure 5 (b). CC-LO achieves lower 99-th percentile latencies only at the lowest load condition (0.35 vs 0.45 msec).

\pvs
~\\\noindent{\bf Throughput.} Contrarian consistently achieves a higher throughput  than CC-LO. Contrarian's maximum throughput is 1.45x CC-LO's in the 1-DC case, and 1.6x in the 2-DC case. In addition, Contrarian achieves a 1.9x throughput improvement when scaling from 1 to 2 DCs. By contrast, CC-LO improves its throughput only by 1.6x. This result is due to the higher replication costs in CC-LO, which has to communicate the dependency list of a replicated update, and perform the readers check in the remote DC.

\pvs
~\\\noindent{\bf Overhead analysis.} To provide a sense of the overhead of the readers check, we present some data collected on the singe-DC platform at the load value at which CC-LO achieves its peak throughput (corresponding to 256 client threads). A readers check targets on average 20 keys, causing the checking partition to contact on average 12 other partitions. A readers check collects on average 252 {\em distinct} ROT ids, which almost matches the number of clients for this experiment. However, the same ROT id can appear in the readers set of multiple keys that have to be checked at different partitions. This increases the {\em cumulative} number of ROT ids exchanged during the readers-check phase, to on average 855 ROT ids for each readers check (71 per contacted node), corresponding roughly to 7KB of data (using 8 bytes per ROT id). Figure~\ref{fig:eval:rdrs} shows that the average overhead of a readers check grows linearly with the number of clients in the system. This result matches our theoretical analysis (see Section~\ref{sec:theory}) and highlights the inherent scalability limitations of latency-optimal ROTs.

\subsection{Effect of write intensity.}
Figure~\ref{fig:write} shows how the performance of the systems is affected by varying the write intensity of the workload.%, on both platforms. 

\begin{figure}[b!]
        \centering
       \includegraphics[scale=0.55]{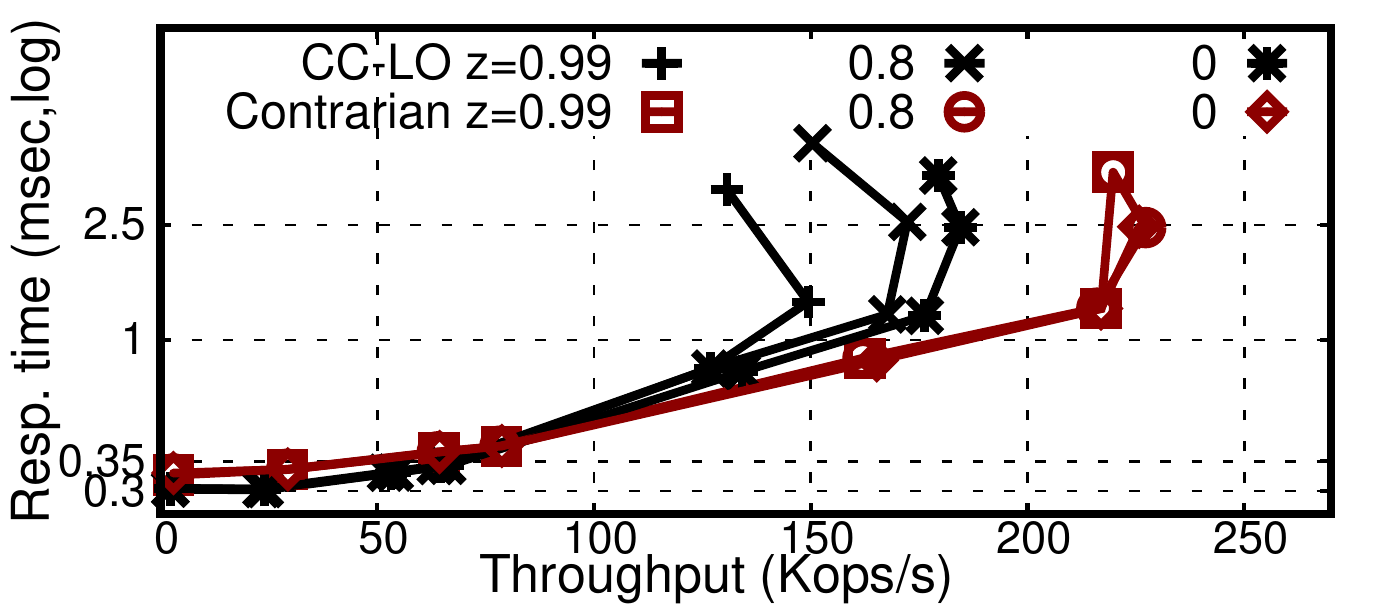}
\caption{Effect of the skew in data popularity (single-DC). Skew hampers the performance of CC-LO, because it leads to long causal dependency chains among operations and thus to much information exchanged during the readers check. 
 }
 \label{fig:zipf}
\end{figure}

\pvs
~\\\noindent{\bf Latency.} Similarly to what is seen previously, for non-trivial load conditions Contrarian achieves lower ROT latencies than CC-LO on both the 1-DC and 2-DC scenarios and with almost all of the write intensity parameters. 
The only exception occurs in the case of the lowest write intensity, and even in this case the differences remain small, especially for the replicated environment.

For w = 0.01 in the single-DC case (Figure~\ref{fig:eval:write:1DC}),  at the lowest load CC-LO achieves an average ROT latency of 0.3 msec vs 0.35 of Contrarian; at high load (200 Kops/s), ROTs in CC-LO completes in 1.11 msec vs 1.33 msec in Contrarian. 
In the 2-DC deployment, however, the latencies achieved by the two systems are practically the same, except for trivial load conditions (Figure~\ref{fig:eval:write:2DC}). 
This change in the relative performances of the two systems is due to the higher replication cost of CC-LO during the readers check, which has to be performed for each update in each DC.

\pvs
~\\\noindent{\bf Throughput.} Contrarian  achieves a higher throughput than CC-LO in almost all scenarios (up to 2.35x for w=0.1 and 2 DCs). The only exception is  the w = 0.01 case in the single DC deployment (where CC-LO achieves a throughput that is 10\% higher). The throughput of Contrarian grows with the write intensity, because PUTs only touch one partition and are faster than ROTs. Instead, higher write intensities hinder the performance of CC-LO, because they cause more frequent execution of the expensive readers check.

\pvs
~\\\noindent{\bf Overhead analysis.} Surprisingly, the latency benefits of CC-LO are not very pronounced, even at the lowest write intensities. This is due to the inherent tension between the frequency of writes and their costs. A low write intensity leads to a low frequency at which readers checks are performed. However, it also means that every write is dependent on many reads, resulting in more costly readers checks.

\subsection{Effect of skew in data popularity.}
Figure~\ref{fig:zipf} depicts how performance varies with skew in data popularity, in the single-DC platform. We focus on this deployment scenario to factor out the replication dynamics of CC-LO and focus on the inherent costs of latency optimality.

\pvs
~\\\noindent{\bf Latency.} Similarly to the previous cases, Contrarian achieves ROT latencies that are lower than CC-LO's for non-trivial load conditions ($> 70$ Kops/s, i.e.,  30\% of Contrarian's maximum throughput).

\pvs
~\\\noindent{\bf Throughput.} The data popularity skew does not sensibly affect Contrarian, whereas it hampers the throughput of CC-LO. The performance of CC-LO degrades because a higher skew causes longer causal dependency chains among operations~\cite{Bailis:2013,Du:2014}, leading to a higher overhead incurred by the readers checks.

\pvs
~\\\noindent{\bf Overhead analysis.} With low skew, a key $x$ is infrequently accessed, so it is likely that many entries in the readers of $x$ can be garbage-collected by the time $x$ is involved in a readers check. With higher skew levels, a few hot keys are accessed most of the time, which leads to the old reader record with many fresh entries. High skew also leads to more duplicates in the ROT ids retrieved from different partitions, because the same ROT id is likely to be present in many the old reader record. Our experimental results (not reported for space constraints) confirm this analysis. They also show that, at any skew level, the number of ROT ids exchanged during a readers check grows linearly with the number of clients (which matches our later theoretical analysis). 

\begin{figure}[b!]
    \centering
       \includegraphics[scale=0.55]{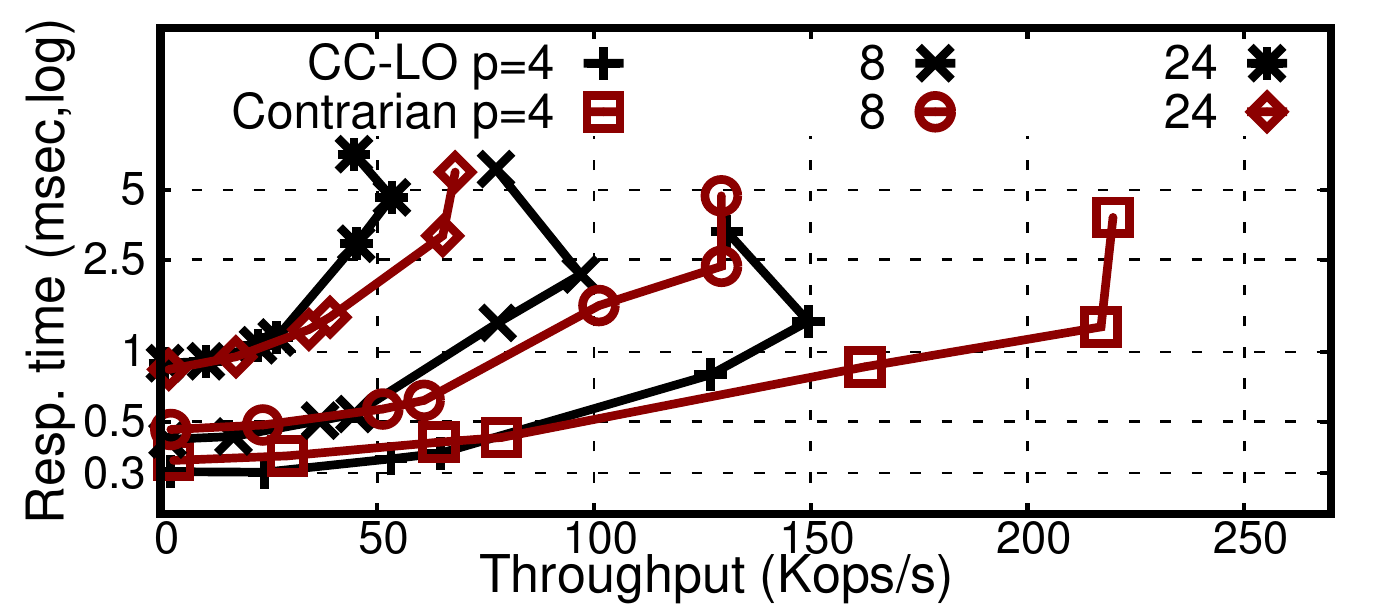}
\caption{Effect of ROT sizes (single-DC). The latency advantage of CC-LO at low load decreases as p grows, because contacting more partitions amortizes the cost of the extra communication needed by Contrarian to execute a ROT.}\label{fig:p}
\end{figure}
\subsection{Effect of size of transactions.} 
Figure~\ref{fig:p} shows the performance of the systems while varying the number of partitions involved in a ROT. We again report results corresponding to the single-DC platform.

\pvs
~\\\noindent{\bf Latency.} Contrarian achieves ROT latencies that are lower than or comparable to CC-LO's for any number of partitions involved in a ROT.  The latency benefits of CC-LO over Contrarian at low load decrease as p grows, because contacting more partitions amortizes the cost of the extra communication needed by Contrarian to execute a ROT.

\pvs
~\\\noindent{\bf Throughput.} Contrarian  achieves higher throughput than CC-LO (up to 1.45x higher, with p=4) for any value of p. %CC-LO achieves marginally lower ROTs latencies than Contrarian only at the lowest load (0.31 vs 0.35 msec in the best case).
The throughput gap between the two systems shrinks with p, because of the extra messages that are sent in Contrarian from the coordinator to the other partitions involved in a ROT. The fact that only one key per partition is read in our experiment is an adversarial setting for Contrarian, because it exacerbates the cost of the extra communication hop used to implement ROTs. Such communication cost would be amortized if ROTs read multiple items per partition. Contrarian can be configured to resort to the 2-round ROT implementation when contacting a large number of partitions, to increase resource efficiency. We are currently testing this optimization. 
 
\subsection{Effect of size of values.}  Larger items naturally result in higher CPU and network costs for marshalling, unmarshalling and transmission operations.  As a result, the performance gap between the systems shrinks as the size of the item values increases. Even in the  case corresponding to large items, however, Contrarian achieves ROT latencies lower than or comparable to the ones achieved by CC-LO, and a 43\% higher throughput (in the single-DC scenario). We omit plots and additional details for space constraints.

\section{Theoretical Results}
\label{sec:theory}
Our experimental study shows that the state-of-the-art CC design for LO ROTs delivers sub-optimal performance, caused by the overhead (imposed on PUTs) for dealing with old readers. One can, however, conceive of alternative implementations. 
% changed on 23 Feb - start 
For instance, rather than storing old readers with the data items in the partitions, 
one could contemplate an implementation which stores old readers at the client which does a PUT and forwards this piece of information to other partitions when doing next PUTs.
% changed on 23 Feb - end
Albeit in a different manner, this implementation still communicates the old readers to the partition where a PUT is performed.
% changed on 23 Feb - start 
One may then wonder: is there an implementation that avoids this overhead altogether in order not to exhibit the performance issues we have seen with CC-LO in Section~\ref{sec:eval}?
% changed on 23 Feb - end

% changed on 23 Feb - start 
We now address this question. We show that the extra overhead on PUTs is {\em inherent} to LO by Theorem 1. 
Furthermore, we show that the extra overhead grows with the number of clients, implying the growth with the number of ROTs and echoing the measurement results we have reported in Section~\ref{sec:eval}.
% changed on 23 Feb - end
% A proof overview starts here
Our proof is by contradiction and consists of three steps. First, we construct a set $\mathcal{E}$ of at least two executions in each of which, different clients issue the same ROT on keys $x,y$ and then causally related PUTs on $x, y$ occur. Our assumption for contradiction is as follows: in our construction, although different clients issue the same ROT, the communication between servers remains the same. (In other words, roughly speaking, servers do not notice all clients that issue the ROT.) Then based on our assumption, we are able to construct another execution $E^*$ in which some clients issue the ROT while causally related PUTs on $x,y$ occur. Finally, still based on our assumption, we show that in $E^*$, although the ROT is in parallel with the causally related PUTs, no server is able to tell so and then the ROT returns a causally inconsistent snapshot. This completes our proof: (roughly speaking) servers must communicate all clients that issue a ROT and the worst-case communication is then linear in the number of clients.
% A proof overview ends here

% changed on 23 Feb - start 
Our theorem applies to the system model described in Section~\ref{sec:model}.
Below we start with an elaboration of our system model (Section~\ref{sec:theory:model}) and the definition of LO (Section~\ref{sec:theory:lo}). 
Then we present and prove our theorem (Section~\ref{sec:theory:theorem}).% and provide a proof (Section~\ref{sec:theory:theorem}). 
% changed on 23 Feb - end 

\subsection{System Model}
\label{sec:theory:model}
For the ease of definitions (as well as proofs), we assume the existence of an accurate real-time clock to which no partition or client has access. When we mention time, we refer to this clock. Furthermore, when we say that two client operations are concurrent, we mean that the duration of the two operations overlap according to this clock.

Among other things, this clock allows us to give a precise definition of eventual visibility.
% changed on 23 Feb - start 
If PUT$(x,X)$ starts at time $T$ (and eventually ends), then there exists finite time $\tau_{X} \geq T$ such that any ROT that reads $x$ and is issued at time $t \geq \tau_X$ returns either $X$ or some $X'$ of which PUT$(x,X')$ starts no earlier than $T$; we say $X$ is \emph{visible} since $\tau_X$.
% changed on 23 Feb - end 

% changed on 23 Feb - start 
We assume the same APIs as described in Section~\ref{sec:model:api}. Clients and partitions exchange messages of which delays are finite, but can be unbounded. Clients and partitions can use their local clocks; however clock drift can be arbitrarily large and infinite (so for some time moment $T$, some clock can never reach $T$).
To capture the design of CC-LO, we also assume that an idle client sends no message to any partition; when
performing an operation on some keys, a client sends messages only to the partitions which store values for these keys; a partition sends messages to client $c$ only when responding to some operation issued by $c$; and clients do not communicate with each other. 
For simplicity, we consider any client issuing a new operation only after its previous operation returns. 
% changed on 23 Feb - end 
% There are four possibilities: R R, R W, W R, W W
% R W and W W cannot do the next before the previous finishes due to dep
% W R can cause problems when defining the return value of o if W and R access the same object o
% R R can be OK, but may create problems to servers.
We assume at least two partitions and a potentially growing number of clients.

\subsection{Properties of LO ROTs}
\label{sec:theory:lo}

\begin{figure*}[ht]
    \centering
    \begin{subfigure}{0.3\textwidth}
       \includegraphics[scale=0.25]{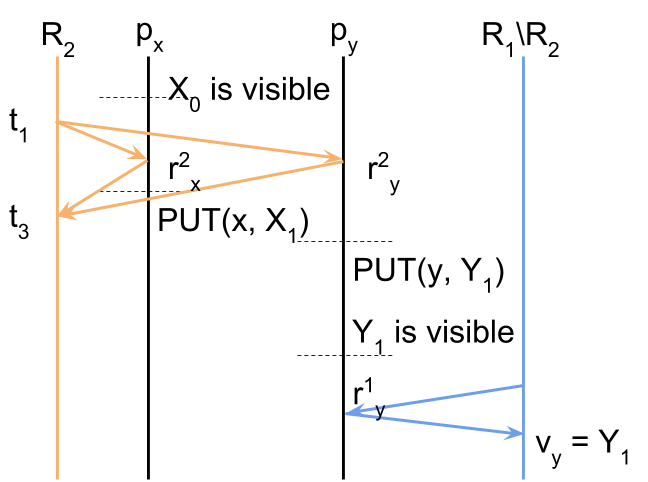}
        \caption{Execution $E_2$}
        \label{fig:e2}
    \end{subfigure}
    \hspace{3cm}
    \begin{subfigure}{0.3\textwidth}
       \includegraphics[scale=0.25]{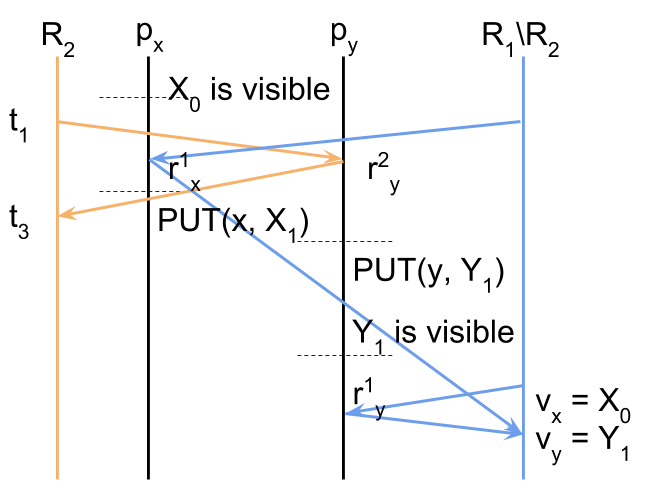}
        \caption{Execution $E^*$ (with $r_x^2$ omitted)}
        \label{fig:e2star}
    \end{subfigure}
\caption{Two (in)distinguishable executions in the proof of Theorem \ref{thm:cost}}\label{fig:pf}
\end{figure*}

% changed on 23 Feb - start 
We adopt the definition of LO ROTs from~\cite{Lu:2016}, which refers to three properties: \emph{one-round}, \emph{one-version}, and \emph{nonblocking}. 
The \emph{one-round} property states that for every client $c$'s ROT $\alpha$, $c$ sends one message to and receives one message from each partition involved in $\alpha$. The \emph{nonblocking} property states that for any partition $p$ to which $c$ sends a message, $p$ eventually sends one message (the one defined in the one-round property) to $c$, even if $p$ receives no message from a server during $\alpha$. A formal definition of \emph{one-version} property is more involved. 
Basically, for every client $c$'s ROT $\alpha$, we consider the maximum amount of information that may be calculated by any implementation algorithm of $c$ based on the messages which $c$ receives during $\alpha$.\footnote{We consider the amount of information instead of the plaintext as values can be encoded in different ways. For example, if a message contains $X_1$ and $X_1\oplus X_2$ for two values $X_1, X_2$ of the same key, then in the plaintext, there is only one version yet some implementation can calculate two versions from the plaintext. The definition of one-version property excludes such message as well as such implementation.}
The \emph{one-version} property specifies that given the messages which $c$ receives from any (non-empty) subset $Par$ of partitions during $\alpha$, the maximum amount of information contains only one version per key for the keys stored in $Par$.
% changed on 23 Feb - end 

\subsection{The cost of $LO$}
\label{sec:theory:theorem}

We say a PUT operation $\alpha$ \emph{completes} if $i)$ $\alpha$ returns to the client that issued $\alpha$; and $ii)$ the value written by $\alpha$ becomes visible. 
Our theoretical result  (Theorem \ref{thm:cost}) highlights  that the cost of $LO$ may occur before any \emph{dangerous} PUT completes. 
% changed on 26 Feb - start 
(We say a PUT operation $\alpha$ is dangerous if $\alpha$ causally depends on some PUT that causally depends on and overwrites a non-$\bot$ value.)
% changed on 26 Feb - end 

\begin{theorem}[Cost of LO ROTs]
\label{thm:cost}
Achieving $LO$ ROT requires communication, potentially growing linearly with the number of clients, before every dangerous PUT completes.
\end{theorem}

% changed on 26 Feb - start 
The intuition behind the cost of $LO$ is that a (dangerous) PUT operation, PUT$(y, Y_1)$, eventually completes; however, due to the asynchronous network, a request resulting from a ROT operation $\alpha$ which reads keys ${x,y}$ may arrive after PUT$(y, Y_1)$ completes, regardless of the other request(s) resulting from $\alpha$. Suppose that $\alpha$ has returned value $X_0$ with respect to value $X_1$ such that 
%NB: we already said that X_i is a version of key x in Section 2
%$X_0$ and $X_1$ are identified by the same key $x$ and 
 $X_0\leadsto X_1\leadsto Y_1$, then $\alpha$ can be at risk of breaking causal consistency. As a result, the partition which provides $X_0$ should notify others of the risk and hence the communication. 
% changed on 26 Feb - end 

% changed on 26 Feb - start 
Inspired by our intuition, we assume that keys $x$, $y$ belong to different partitions $p_x$ and $p_y$, respectively. 
% changed on 26 Feb - end 
We call client $c$ an \emph{old reader} of $x$, with respect to PUT$(y, Y_1)$,\footnote{
The definition of an old reader of $x$ here specifies a certain PUT on $y$ and is thus more specific than the definition in CC-LO, an old reader of $x$ in general. The reason to specify a certain PUT is to emphasize the causal relation $X_1\leadsto Y_1$. The proof hereafter takes the more specific definition when mentioning old readers.} if $c$ issues a ROT operation which (1) is concurrent with PUT$(x, X_1)$ and PUT$(y, Y_1)$ and (2) returns $X_0$.
In general, if $c$ issues a ROT operation that reads $x$, then we say $c$ is a reader of $x$. Thus the risk lies in the fact that due to the asynchronous network, any reader can potentially be an old reader.

% changed on 26 Feb - start 
To have $X_0\leadsto X_1\leadsto Y_1$, for simplicity, we consider a scenario where some client $c_w$ does four PUT operations in the following order: 
PUT$(x, X_0)$, PUT$(y, Y_0)$, PUT$(x, X_1)$ and PUT$(y, Y_1)$, and $c_w$ issues each PUT (except for the first one) after the previous PUT completes.
To prove Theorem \ref{thm:cost}, we consider the worst case: 
all clients except $c_w$ can be readers. We identify similar executions where a different subset of clients are readers. Let $\mathcal{D}$ be the set of all clients except $c_w$. 
We construct the set $\mathcal{E}$ such that each execution has one subset of $\mathcal{D}$ as readers. Hence $\mathcal{E}$ contains $2^{|\mathcal{D}|}$ executions in total. We later show one execution in $\mathcal{E}$ in which the communication carrying readers grows linearly with $|\mathcal{D}|$ and thus prove Theorem \ref{thm:cost}. 
% changed on 26 Feb - end

%%\pvs
~\\
\noindent\textbf{$2^{|\mathcal{D}|}$ executions $\mathcal{E}$.}
% changed on 26 Feb - start 
Each execution $E\in\mathcal{E}$ is based on a subset $R$ of $\mathcal{D}$ as readers. 
% changed on 26 Feb - end
Every client $c$ in $R$ issues ROT$(\{x, y\})$ at the same time $t_1$. By one-round property, $c$ sends two messages $m_{x, req}$, $m_{y,req}$ to $p_x$ and $p_y$ respectively at $t_1$.
We denote the event that $p_x$ receives $m_{x, req}$ by $r_x$, the event that $p_y$ receives $m_{y, req}$ by $r_y$.
By the nonblocking property, $p_x$ and $p_y$ can be considered to receive messages from $c$ and send messages to $c$ at the same time $t_2$.\footnote{Clearly, $p_x$ and $p_y$ may receive messages at different time, and the proof still holds. The same time $t_2$ is assumed for the simplicity of presentation.} Finally, $c$ receives messages from $p_x$ and $p_y$ at the same time $t_3$.
We order events as follows: $X_0$ and $Y_0$ are visible, $t_1$, $r_x = r_y = t_2$, PUT$(x, X_1)$ is issued, $t_3$, PUT$(y, Y_1)$ is issued. Let $\tau_{Y_1}$ be the time when PUT$(y, Y_1)$ completes. 
% changed on 26 Feb - start 
For every execution in $\mathcal{E}$, $t_1, t_2, t_3$ take the same values while $\tau_{Y_1}$ actually denotes the maximum value. 
% changed on 26 Feb - end

% changed on 26 Feb - start 
To emphasize the burden on $p_y$, we consider communication that \emph{precedes} a message that $p_y$ receives: we say message $a$ precedes message $b$ if (1) some process $p$ sends $b$ after $p$ receives $a$, or (2) $\exists$ message $c$ such that $a$ precedes $c$ and $c$ precedes $b$.
Clearly, the executions in $\mathcal{E}$ are the same until time $t_1$. Since $t_1$, these executions, especially, the communication between $p_x$ and $p_y$ may change. We construct all executions in $\mathcal{E}$ altogether: 
if at some time point, in one execution, some server sends a message, then we construct all other executions such that the same server sends the same message except that the server is $p_x$, $p_y$ or contaminated by $p_x$ or $p_y$. By contamination, we mean that at some point, $p_x$ or $p_y$ sends message $m$ but we are unable to construct all other executions to do the same; then the message $m$ and server $s$ which receives $m$ are contaminated and $s$ can further contaminate other servers. In our construction, we focus on the non-contaminated messages which are received at the same time across all executions in $\mathcal{E}$. For other messages, if in at least two executions, the same contaminated message $m$ can be sent, then we let $m$ to be received at the same time across these executions; otherwise, We do not restrict the schedule.
% changed on 26 Feb - end
% further changed on 1 Mar

% changed on 26 Feb - start 
We show that the worst-case execution exists in our construction of $\mathcal{E}$. To do so, we first show a property of $\mathcal{E}$; i.e., for any two executions $E_1$, $E_2$ in $\mathcal{E}$ (with different readers), the communication of $p_x$ and $p_y$ must be different, as formalized in Lemma \ref{lma:ineq}.\footnote{Lemma \ref{lma:ineq} abstracts ways of communication between $p_x$ and $p_y$ so that it is independent of certain implementations, 
and covers the following example implementations of communication for old readers as in CC-LO, as the example introduced at the beginning of this section, as well as the following: $p_y$ keeps asking $p_x$ whether a reader of $Y_0$ is a reader of $X_0$ to determine whether all readers of $X_0$ have arrived at $p_y$ (so that there is no old reader with respect to $Y_1$).}
% changed on 26 Feb - end

\begin{lemma}[Different readers, different messages]
\label{lma:ineq}
Consider any two executions $E_1, E_2\in\mathcal{E}$.
In $E_i, i\in\{1,2\}$, denote by $M_i$ the messages which $p_x$ or $p_y$ sends to a process other than $\mathcal{D}$ and which precedes some message that $p_y$ receives during $[t_1, \tau_{Y_1}]$ in $E_i$, and denote by $str_i$ the concatenation of ordered messages in $M_i$ ordered by the time when every message is sent. Then $str_1\neq str_2$.
\end{lemma}

%\pvs
The main intuition behind the proof is that if communication were the same regardless of readers, $p_Y$ would be unable to distinguish readers from \emph{old} readers.
Suppose now by contradiction that $str_1 = str_2$. 
% changed on 26 Feb - start 
Then our construction of $\mathcal{E}$ allows us to construct an special execution $E^*$ based on $E_2$ (as well as $E_1$).
% changed on 26 Feb - end
Let the subset of $\mathcal{D}$ for $E_i$ be $R_i$ for $i\in\{1,2\}$. W.l.o.g., $R_1\backslash R_2\neq \emptyset$. We construct $E^*$ such that clients in $R_1\backslash R_2$ are old readers 
% changed on 26 Feb - start 
(and show that $E^*$ breaks causal consistency due to old readers).
% changed on 26 Feb - end

\pvs
~\\
\noindent\textbf{Execution $E^*$ with old readers.}
% changed on 26 Feb - start 
In $E^*$, both $R_1$ and $R_2$ issue ROT$(\{x,y\})$ at $t_1$. 
% changed on 26 Feb - end
To distinguish between events (as well as messages) resulting from $R_1$ and $R_2$, we use superscripts $1$ and $2$ to denote the events, respectively. 
% changed on 26 Feb - start 
For simplicity of notations, in $E_2$, we call the two events at the server-side (for which $p_x$ and $p_y$ receive messages from $R_2$ respectively) also $r_x^2$ and $r_y^2$, illustrated in Figure \ref{fig:e2}.
% changed on 26 Feb - end
In $E^*$, we now have four events at the server-side: $r_x^1$, $r_y^1$, $r_x^2$, $r_y^2$. 
We construct $E^*$ based on $E_2$ by scheduling $r_x^1$ and $r_y^2$ in $E^*$ at $t_2$ (the same time as $r_x^2$ and $r_y^2$ in $E_2$), and postponing $r^1_y$ (as well as $r_x^2$), as illustrated in Figure \ref{fig:e2star}. 
The ordering of events in $E^*$ is thus different from $E_2$. More specifically, the order is: $X_0$ and $Y_0$ are visible, $t_1$, $r^1_x = r^2_y = t_2$, PUT$(x, X_1)$ is issued, PUT$(y, Y_1)$ is issued, $\tau_{Y_1}$, $r^1_y$ (for every client in $R_1\backslash R_2$ as $r^2_y$ has occurred),  $r^2_x$ (for every client in $R_2\backslash R_1$, not shown in Figure \ref{fig:e2star}), 
$R_1\backslash R_2$ returns ROT. By asynchrony, the order is legitimate, which results in old readers $R_1\backslash R_2$.

\begin{proof}[Proof of Lemma \ref{lma:ineq}]
% changed on 26 Feb - start
Our proof is by contradiction.
As $str_1 = str_2$, according to our construction, $p_y$ does not receive any message preceded by some different contaminated message in $E_1$ and $E_2$. Therefore even if we replace $r_x^2$ in $E_2$ for $r_x^1$ in $E^*$ (as in $E_1$), then by
% changed on 26 Feb - end
$\tau_{Y_1}$, $p_Y$ is unable to distinguish between $E_2$ and $E^*$.

Previously, our construction of $E_2$ is until $\tau_{Y_1}$. Let us now extend $E_2$ so that $E_2$ and $E^*$ are the same after $\tau_{Y_1}$. Namely, in $E_2$, after $\tau_{Y_1}$, every client $c_1\in R_1\backslash R_2$ issues ROT$(\{x, y\})$; and as illustrated in Figure \ref{fig:pf}, $r^1_y$ is scheduled at the same time in $E_2$ and in $E^*$.

Let $\vec{v}$ be the return value of $c_1$'s ROT in either execution. By eventual visibility, in $E_2$, $v_y = Y_1$. 
We now examine $E^*$. By eventual visibility, as $t_1$ is after $X_0$ and $Y_0$ are visible, $v_x, v_y\neq \bot$. 
As $r^1_x$ is before PUT$(x, X_1)$ is issued, $v_x\neq X_1$. By $p_y$'s indistinguishability between $E_2$ and $E^*$,  and according to the one-version property, $v_y = Y_1$ as in $E_2$. 
Thus in $E^*$, $v_x = X_0$ and $v_y = Y_1$, a snapshot that is not causally consistent. A contradiction.
\end{proof}

Lemma \ref{lma:ineq} demonstrates a property for any two executions in $\mathcal{E}$, which implies another property of $\mathcal{E}$: if for any two executions, communication has to be different, then for all executions, the number of possibilities of what is communicated grows with the number of elements in $\mathcal{E}$. Recall that $|\mathcal{E}|$ is a function of $|\mathcal{D}|$. Hence, we connect the communication and $|\mathcal{D}|$ in Lemma \ref{lma:lower-bound}.

%\pvs

\begin{lemma}[Lower bound on the cost]
\label{lma:lower-bound}
Before PUT$(y, Y_1)$ completes, in at least one execution in $\mathcal{E}$, the communication of $p_x$ and $p_y$ takes at least $\mathcal{L}(|\mathcal{D}|)$ bits where $\mathcal{L}$ is a linear function.
\end{lemma}

\begin{proof}[Proof of Lemma \ref{lma:lower-bound}]
We index each execution $E$ by the set $R$ of clients which issue ROT$(\{x,y\})$ at time $t_1$. We have therefore $2^{|\mathcal{D}|}$ executions: $\mathcal{E} = \{E(R)|R\subseteq\mathcal{D}\}$. Let $b(R)$ be the messages which $p_x$ and $p_y$ send in $E(R)$ as defined in Lemma \ref{lma:ineq}, and let $B = \{b(R)|R\subseteq\mathcal{D}\}$. By Lemma \ref{lma:ineq}, we can show that $\forall b_1, b_2\in B, b_1\neq b_2$. 
Then $|B| = |\mathcal{E}| = 2^{|\mathcal{D}|}$.
Therefore, it is impossible that every element in $B$ has fewer than $|\mathcal{D}|$ bits. In other words, 
in $\mathcal{E}$, we have at least one execution $E = E(R)$  where $b(R)$ takes at least $\log_2(2^{|\mathcal{D}|}) = |\mathcal{D}|$ bits, a linear function in $|\mathcal{D}|$.
\end{proof}

Recall that $|\mathcal{D}|$ is a variable that grows linearly with the number of clients.
Thus following Lemma \ref{lma:lower-bound}, we find $\mathcal{E}$ contains a worst-case execution that supports Theorem \ref{thm:cost} and we thus complete the proof of Theorem \ref{thm:cost}.

\pvs
~\\
\noindent\textbf{Remark on implementations.}
The proof shows the necessary communication of readers when each client issues one operation. 
Here we want to make the link back to the implementation of LO ROTs in CC-LO. The reader may wonder in particular about the relationship between the transaction identifiers that are sent as old readers in CC-LO, and the worst-case communication linear in the number of clients derived in the theorem. In fact, the CC-LO implementation considers that clients may issue multiple transactions at the same time, and then different ROTs of a single client should be considered as different readers, hence the use of transaction identifiers to distinguish one from another.

A final comment is on a straw-man implementation where each operation is attached to the output of a Lamport Clock \cite{Lamport:1978} (called logical time below) alone. 
% changed on 26 Feb - start 
Such implementation (without communication of potentially old readers) still fails. 
The problem is that the number of increments in logical time after ROTs is at most the number of all ROTs, i.e., $|\mathcal{D}|$. Then for some $E_1$ and $E_2$, Lemma \ref{lma:ineq} does not hold, i.e., the communication is the same.  Although when issuing the ROT, client $c$ in $R_1\backslash R_2$ can send logical time to servers, the logical time sent in $E_2$ and $E^*$ is the same and thus does not help $p_y$ to distinguish between $E_2$ and $E^*$, resulting in the violation of causal consistency again. Hence communication of readers, as Theorem \ref{thm:cost} indicates, is still required for this straw-man implementation.
% changed on 26 Feb - end
\section{Related work}
\label{sec:rw}
\begin{table*}[ht]
\centering
\begin{tabular}{lcccccccc}
\hline
\multicolumn{1}{|l|}{\multirow{3}{*}{System}} & \multicolumn{3}{c|}{ROT latency optimality}                                                                                                                         & \multicolumn{4}{c|}{Write cost}                                                                                                   & \multicolumn{1}{c|}{\multirow{3}{*}{Clock}} \\ \cline{2-8}
\multicolumn{1}{c|}{}                        & \multicolumn{1}{c|}{\multirow{2}{*}{Nonblocking}} & \multicolumn{1}{c|}{\multirow{2}{*}{\#Rounds}} & \multicolumn{1}{c|}{\multirow{2}{*}{\#Versions}} & \multicolumn{2}{c|}{Communication}                                  & \multicolumn{2}{c|}{Meta-data}                              & \multicolumn{1}{c|}{}                            \\ \cline{5-8}
\multicolumn{1}{|c|}{}                        & \multicolumn{1}{l|}{}                          & \multicolumn{1}{l|}{}                       & \multicolumn{1}{l|}{}                            & \multicolumn{1}{l|}{$c\leftrightarrow s$} & \multicolumn{1}{c|}{$s\leftrightarrow s$} & \multicolumn{1}{l|}{$c\leftrightarrow s$} & \multicolumn{1}{l|}{s$\leftrightarrow $s} & \multicolumn{1}{l|}{}                            \\ \hline 
COPS~\cite{Lloyd:2011}                                          & \cmark                                         & $\leq 2$                                      & $\leq 2$                                           & 1                            & -                            & $|$deps$|$                   & -                            & Logical                                          \\
Eiger~\cite{Lloyd:2013}                                         & \cmark                                         & $\leq 2$                                      & $\leq 2$                                           & 1                            & -                            & $|$deps$|$                   & -                            & Logical                                          \\
ChainReaction~\cite{Almeida:2013}                                 & \xmark                                         & $\geq$ 2                                    & 1                                           & 1                            & $\geq$ 1                            & $|$deps$|$                   & M                         & Logical                                          \\
Orbe~\cite{Du:2013}                                          & \xmark                                              & 2                                           & 1                                                & 1                            & -                            & NxM                       & -                            & Logical                                          \\
GentleRain~\cite{Du:2014}                                    & \xmark                                         & 2                                         & 1                                           & 1                            & -                            & 1                         & -                            & Physical                                         \\
Cure~\cite{Akkoorath:2016}                                          & \xmark                                         & 2                                         & 1                                           & 1                            & -                            & M & -                            & Physical                                         \\
OCCULT$^\dagger$~\cite{Mehdi:2017}                                        & \cmark                                         & $\geq$ 1                                    & $\geq$1                                          & 1                            & -                            & O(P)                         & -                            & Hybrid                                           \\
POCC~\cite{Spirovska:2017}                                     & \xmark                                         & 2                                      & 1                                           & 1                            & -                         & M                   & -                         & Physical                                         \\
COPS-SNOW~\cite{Lu:2016}                                     & \cmark                                         & 1                                      & 1                                           & 1                            & O(N)                         & $|$deps$|$                   & O(K)                         & Logical                                          \\
\hline
\hline
{\bf Contrarian}                                    & \cmark                                         & 1 1/2 (or 2)                                        & 1                                           & 1                            & -                            & M                         & -                            & Hybrid                                          
\end{tabular}
\caption{Characterization of CC systems with ROTs support, in a geo-replicated setting. N, M and K represent, respectively, the number of partitions, DCs, and clients in a DC. $\dagger$ indicates a single-master system, and $P$ represents the number of DCs that act as master for at least one partition. $c \leftrightarrow s$, resp., $s \leftrightarrow s$, indicates client-server, resp. inter-server, communication. 
}\label{tab:rw}
\end{table*}
\noindent{\bf Causally consistent systems.} Table~\ref{tab:rw} classifies existing systems with ROT support according to the cost of performing ROT and PUT operations.
COPS-SNOW is the only latency-optimal system. COPS-SNOW achieves latency optimality at the expense of more costly writes, which carry detailed dependency information and incur extra communication overhead. 
Previous systems fail to achieve at least one of the sub-properties of latency optimality.

ROTs in COPS and Eiger might require two rounds of client-server communication to complete. The second round is needed if the client reads, in the first round, two items that might belong to different causally consistent snapshots. COPS and Eiger rely on fine-grained protocols to track and check the dependencies of replicated updates (see Section~\ref{sec:sys}), which have been shown to limit their scalability~\cite{Du:2013,Du:2014,Akkoorath:2016}. ChainReaction uses a potentially-blocking and potentially multi-round protocol based on a per-DC {\em sequencer} node. % that assigns timestamps to PUT and ROT operations. 
 
 Orbe, GentleRain, Cure and POCC use a coordinator-based approach similar to what described in Section~\ref{sec:sys}, and require two communications rounds. These systems use physical clocks and may block ROTs either because of clock skew or to wait for the receipt of remote updates.

Occult uses a primary-replica approach and use HLCs to avoid blocking due to clock skew. Occult implements ROTs that run in potentially more than one round and that potentially span multiple DCs (which makes the system not always-available). Occult requires at least one dependency timestamp for each DC that hosts a master replica.

Unlike these systems, Contrarian leverages HLCs to implement ROTs that are always-available, nonblocking and always complete in 1 1/2 (or 2) rounds of communication.

%\end{comment}

Other %relevant
CC systems include SwiftCloud~\cite{Zawirski:2015}, 
Bolt-On~\cite{Bailis:2013}, Saturn~\cite{Bravo:2017}, Bayou~\cite{Petersen:1997,Terry:1995}, PRACTI~\cite{Belaramani:2006}, ISIS~\cite{Birman:1987}, lazy replication~\cite{Ladin:1992}, causal memory~\cite{Ahamad:1995}, EunomiaKV~\cite{Eunomia:2017} and CausalSpartan~\cite{CausalSpartan}. These systems either do not support ROTs, or  target a different model from the one considered in this paper, e.g., they do not implement sharding the data set in partitions. Our theoretical results require at least two partitions. Investigating the cost of LO in other system models is an avenue for future work.

CC is also implemented by systems that support different consistency levels~\cite{Crooks:2016}, implement strong consistency on top of CC~\cite{Balegas:2015}, and combine different consistency levels depending on the semantics of operations~\cite{Li:2014,Balegas:2016} or on target SLAs~\cite{Ardekani:2014,Terry:2013}. Our theorem provides a lower bound on the overhead of latency-optimal ROTs with CC. Hence, any system that implements CC or a strictly stronger consistency level cannot avoid such overhead. We are investigating how the lower bound on this overhead varies depending on the consistency level, and what is its  effect on performance.

\pvs
~\\\noindent{\bf Theoretical results on causal consistency.} Causality was introduced by Lamport~\cite{Lamport:1978}. Hutto and Ahamad~\cite{Hutto:1990} provided the first definition of causal consistency, later revisited from different angles~\cite{Mosberger:1993,Adya:1999,Crooks:2016,Viotti:2016}. Mahajan et al. have proved that real-time CC is the strongest consistency level that can be obtained in an always-available and one-way convergent system~\cite{Mahajan:2011}. Attiya et al. have introduced the observable CC model and have shown that it is the strongest that can be achieved by an eventually consistent data store implementing multi-value registers~\cite{Attiya:2015}.

The SNOW theorem~\cite{Lu:2016} shows that LO can be achieved by any system that $i)$ is not strictly serializable~\cite{Papadimitriou:1979} or $ii)$ does not support write transactions. Based on this result, the SNOW paper suggests that any protocol that matches one of these two conditions can be {\em improved} to be latency-optimal. The SNOW paper {\em indicates} that a way to achieve this is to shift the overhead from ROTs to writes. 
In this paper, we {\em prove} that achieving latency optimality in CC implies an extra cost on writes, which is inherent and significant. 

Bailis et al. study the overhead of replication and dependency tracking in geo-replicated CC systems~\cite{Bailis:2012b}. By contrast, we investigate the inherent cost of latency-optimal CC designs, i.e., even in absence of (geo-)replication.

\section{Conclusion}
\label{sec:conclusion}

Causally consistent read-only transactions are an attractive primitive for large-scale systems, as 
 they eliminate a number of anomalies and facilitate the task of developers.
Furthermore, given that most applications are expected to be read-dominated, low latency of read-only transactions is of paramount importance to overall system performance.
It would therefore appear that  {\em latency-optimal} read-only transactions, which provide a nonblocking, single-version and single-round implementation, are  particularly appealing.
The catch is that these latency-optimal protocols impose an overhead on writes that is so high that it jeopardizes performance, even in read-heavy workloads.

In this paper, we present an ``almost latency-optimal'' protocol that maintains the nonblocking and one-version aspects of their latency-optimal counterparts, but sacrifices the one-round property and instead runs in one and a half rounds. On the plus side, however, this protocol avoids the entire overhead that latency-optimal protocols impose on writes. As a result, measurements show that this ``almost latency-optimal'' protocol outperforms latency-optimal protocols, not only in terms of throughput, but also in terms of latency, for all but the lowest loads and the most read-heavy workloads. 

In addition, we show that the overhead of the latency-optimal protocol is inherent. In other words, it is not an artifact of current implementations. In particular, we show that this overhead grows linearly with the number of clients.

\bibliographystyle{acm}
\bibliography{biblio}

\end{document}